\documentclass[a4paper,USenglish,cleveref,thm-restate,numberwithinsect,pdfa]{lipics-v2021}

\hideLIPIcs

\nolinenumbers

\crefname{observation}{Observation}{Observations}

\title{Exploring Minimal Blocking Sets for Kernelization Lower Bounds for \texorpdfstring{$\mathcal{F}$}{F}-Minor Deletion}
\titlerunning{Exploring Minimal Blocking Sets for Kernelization Lower Bounds}

\author{Ahmed Ghazy}{CISPA Helmholtz Center for Information Security, Saarbrücken, Germany \and Saarland University, Saarbrücken, Germany}{ahmed.ghazy@cispa.de}{https://orcid.org/0009-0009-7414-5871}{}
\author{Jakob Greilhuber}{CISPA Helmholtz Center for Information Security, Saarbrücken, Germany \and Saarland University, Saarbrücken, Germany}{jakob.greilhuber@cispa.de}{https://orcid.org/0009-0001-8796-6400}{}
\author{Tim A.\ Hartmann}{CISPA Helmholtz Center for Information Security, Saarbrücken, Germany}{hartmann@algo.rwth-aachen.de}{https://orcid.org/0000-0002-1028-6351}{}

\authorrunning{A.~Ghazy, J.~Greilhuber, and T.~A.~Hartmann}

\Copyright{Ahmed Ghazy, Jakob Greilhuber, Tim A.\ Hartmann}



\relatedversion{This is the full version of an IPEC 2026 contribution.}

\acknowledgements{We thank Marin Bougeret, Eric Brandwein, and Ignasi Sau for helpful and constructive communication.}

\usepackage[noadjust]{cite}
\usepackage[skins]{tcolorbox}

\newcommand{\np}{\mathsf{NP}}
\newcommand{\conppoly}{\mathsf{coNP}/\mathsf{poly}}

\newcommand{\range}[2][1]{[#1,#2]}

\newcommand{\problemname}[1]{\textnormal{\textsc{#1}}}
\newcommand{\upsf}[1]{\textnormal{\textsf{#1}}}

\newcommand{\BB}{\mathcal{B}}
\newcommand{\CC}{\mathcal{C}}
\newcommand{\FF}{\mathcal{F}}
\newcommand{\GG}{\mathcal{G}}
\newcommand{\HH}{\mathcal{H}}
\newcommand{\SSS}{\mathcal{S}}
\newcommand{\TT}{\mathcal{T}}
\newcommand{\XX}{\mathcal{X}}
\newcommand{\LL}{\mathcal{L}}
\newcommand{\RR}{\mathcal{R}}

\DeclareMathOperator{\pw}{pw}
\DeclareMathOperator{\tw}{tw}
\DeclareMathOperator{\leafblock}{LeafBlocks}

\newcommand{\defparproblem}[4]{\begin{tcolorbox}[
            enhanced,
            colback=white,
            colframe=black,
            boxrule = 0.5pt,
            coltitle=black,
            title=#1,
            rounded corners,
            attach boxed title to top left={yshift=-10pt, xshift=6pt},
            bottom=-0.5mm,
            boxed title style={
                    interior style={fill=white},
                    frame hidden
                }
        ]\begin{tabularx}{12.5cm}{l r X p {0.5cm}}
             & Input:     & #2 \\
             & Parameter: & #3 \\
             & Question:  & #4
        \end{tabularx}
    \end{tcolorbox}
}

\newcommand{\twEtaDeletion}{\problemname{Tw-$\eta$-Deletion}}

\newcommand{\vcDistToG}{\problemname{Vertex Cover}/\textnormal{\texttt{dist-to-$\mathcal{G}$}}}

\newcommand{\pwEtaDeletion}{\problemname{Pw-$\eta$ Deletion}}

\newcommand{\fDeletion}{\problemname{$\mathcal{F}$-Minor Deletion}}

\newcommand{\pwEtaDeletionEtaPlusOne}{\problemname{Pw-$\eta$ Deletion}/\textnormal{\texttt{dist-to-$\pwLeqEtaClass[\eta+1]$}}}

\newcommand{\fDeletionDistToG}[1][$\GG$]{$\FF$-\problemname{Minor Deletion}/\textnormal{\texttt{dist-to-#1}}}

\newcommand{\OPT}{\textnormal{\texttt{OPT}}}
\newcommand{\opt}{\OPT}

\newcommand{\prune}[2]{#1\textnormal{\textsf{-prune}}(#2)}

\newcommand{\pwLeqEtaClass}[1][\eta]{\GG_{\pw \leq #1}}

\begin{document}

\maketitle

\begin{abstract}
  For \textsc{Vertex Cover}, a minimal blocking set of a graph $G$ is an inclusion-wise minimal vertex set contained in no minimum vertex cover of $G$.
Bougeret et al.~[SIDMA 2022] showed that, for all minor-closed graph classes $\GG$, \textsc{Vertex Cover} has a polynomial kernel parameterized by the distance to $\GG$ if and only if $\GG$ has bounded minimal blocking sets.\footnote{We assume $\np \not \subseteq \conppoly$ throughout the abstract and introduction.}

In this work, we explore how far this connection between minimal blocking sets and kernelization extends to \fDeletion{} parameterized by the deletion distance to $\GG$.
For a fixed nonempty finite family $\mathcal{F}$ of graphs, given a graph $G$ and integer $k$, \fDeletion{} asks whether $G$ can be made free of minors from $\FF$ by deleting at most $k$ vertices.

A minimal $\FF$-blocking set of $G$ is an inclusion-wise minimal vertex set not contained in any optimal \fDeletion{} solution of $G$.
For families $\mathcal{F}$ of connected graphs, we show that if some graph class $\GG$ closed under disjoint union has unbounded minimal $\FF$-blocking sets, then \fDeletion{} parameterized by the deletion distance to $\GG$ has no polynomial kernel.
This generalizes a result of Hols et al.~[SIDMA 2022] for \textsc{Vertex Cover}.

We then show that this criterion does not capture all kernelization lower bounds for \fDeletion{}, even for minor-closed graph classes.
Motivated by the work of Dekker and Jansen~[Discrete Appl. Math. 2024], we study the more general notion of extended $\FF$-blocking sets.
Under additional conditions, we show that $\GG$ having unbounded minimal extended $\FF$-blocking sets rules out polynomial kernels.

For every fixed integer $\eta \geq 1$, \pwEtaDeletion{} is the problem of deleting at most $k$ vertices to reduce the pathwidth of a given graph to at most $\eta$.
We show that it has no polynomial kernel parameterized by the distance to graphs with pathwidth at most $\eta+1$.
Moreover, for all families $\FF$ of graphs, each connected and having at least two vertices, we show that \fDeletion{} parameterized by the distance to graphs with pathwidth at most $\max_{F \in \FF} \pw(F) + 1$ has no polynomial kernel.
These lower bounds rely in part on a pathwidth analysis of a minor modification of a construction due to Donkers and Jansen [J. Comput. Syst. Sci. 2021]. \end{abstract}

\newpage

\section{Introduction}
\label{sec:introduction}
\subparagraph*{Kernelization.}
Kernelization~\cite{cyganParameterizedAlgorithms2015,fominKernelizationTheoryParameterized2019} is the formal study of preprocessing algorithms for parameterized problems.
Given an instance $(I,k)$ of a parameterized problem (with parameter $k$), the goal is to output an equivalent instance $(I',k')$ in polynomial time, where both $|I'|$ and $k'$ are bounded by $f(k)$ for some computable function $f$.
If $f$ is a polynomial function, then such an algorithm is called a \emph{polynomial kernel}.
It is well-known that a decidable parameterized problem is fixed-parameter-tractable (\textsf{FPT}) if and only if it has a kernel, but not necessarily one of polynomial size.
One of the main goals in the field is understanding which problems admit polynomial kernels.

\subparagraph*{Vertex cover.}
\problemname{Vertex Cover} is perhaps the most-studied problem in this paradigm.
There are many fundamental and different techniques to obtain polynomial kernels for the problem when parameterizing by the solution size~\cite{bussNondeterminism1993,fominKernelizationTheoryParameterized2019,cyganParameterizedAlgorithms2015,fellowsBlowupsWinWins2003,chorLinearKernelsLinear2004,chenVertexCoverFurther2001,lampisKernelOrder22011}.
When going beyond the solution size as the parameter, a next natural step is considering the (vertex-deletion) distance to a graph class $\GG$ as the parameter, such that the problem is polynomial-time-solvable on $\GG$.
This follows the paradigm of parameterizing by ``distance from triviality''~\cite{JansenVertexCoverKernelization2013,DBLP:conf/iwpec/GuoHN04,bougeretKernelizationDichotomiesHitting2026STACS}.

Formally, a problem instance of \vcDistToG{} consists of a graph $G$, integer $k$, and set $M \subseteq V(G)$ (a \emph{modulator} to $\GG$) such that $G - M \in \mathcal{G}$.
The parameter is $|M|$, and the question is whether $G$ has a vertex cover of size at most $k$.
One of the most important works in this direction is the result of Jansen and Bodlaender~\cite{JansenVertexCoverKernelization2013} that shows the problem admits a polynomial kernel when $\GG$ is the graph class of forests.
This result has led to an array of further kernelization results for \vcDistToG{}~\cite{bougeretHowMuchDoes2019,cyganHardnessLosingWidth2014,fominVertexCoverStructural2016,hols2022elimination,majumdarPolynomialKernelsVertex2018} showing the existence and non-existence of polynomial kernels for the considered parameterization.
Eventually, the question for which minor-closed graph classes $\mathcal{G}$ the problem admits a polynomial kernel was settled: Bougeret et al.~\cite{bougeretBridgedepthCharacterizesWhich2022} show that, when $\mathcal{G}$ is minor-closed, \vcDistToG{} admits a polynomial kernel if and only if $\mathcal{G}$ has bounded bridgedepth.

\subparagraph*{Minimal blocking sets.}
There is a central concept which kept (sometimes implicitly) appearing in these works: \emph{minimal blocking sets}.
For \textsc{Vertex Cover}, a blocking set of a graph $G$ is a set $X \subseteq V(G)$ such that no minimum vertex cover of $G$ contains all of $X$.
A blocking set is minimal if no proper subset of it is a blocking set.
See \cref{fig:blocking_set_example} for a small example.
Already the work of Jansen and Bodlaender~\cite{JansenVertexCoverKernelization2013} exhibits that bounding the size of minimal blocking sets is useful for kernelization, and shows that any minimal blocking set of a forest has size at most $2$.
The same concept similarly plays a central role in subsequent papers about the problem~\cite{hols2022elimination,bougeretBridgedepthCharacterizesWhich2022,fominVertexCoverStructural2016,holsSmallerParametersVertex2017,majumdarPolynomialKernelsVertex2018,DBLP:conf/wg/AntipovK25,bougeretHowMuchDoes2019}.\footnote{The parameterized complexity of computing the maximum size of a
minimal blocking set has also been studied~\cite{DBLP:journals/algorithmica/AraujoBCS23}.}

\suppressfloats[t]
\begin{figure}[t]
	\centering
	\begin{subfigure}[t]{0.31\linewidth}
		\centering
		\includegraphics[page=1]{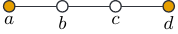}
		\subcaption{Candidate $X$ (orange).}
	\end{subfigure}
	\hfill
	\begin{subfigure}[t]{0.31\linewidth}
		\centering
		\includegraphics[page=2]{figures/blocking-set-example.pdf}
		\subcaption{Minimum cover $\{a,c\}$.}
	\end{subfigure}
	\hfill
	\begin{subfigure}[t]{0.31\linewidth}
		\centering
		\includegraphics[page=3]{figures/blocking-set-example.pdf}
		\subcaption{Minimum cover $\{b,d\}$.}
	\end{subfigure}
	\caption{A minimal blocking set for \textsc{Vertex Cover} ($\FF = \{K_2\}$) on the path $a,b,c,d$. Unselected vertices of $X$ are orange, selected cover vertices are red, and dashed edges are incident to a selected vertex.
	No minimum cover contains both $a$ and $d$, so $X = \{a,d\}$ is a blocking set; but the displayed covers show that $a$ and $d$ each lie in some minimum cover, so $X$ is a minimal blocking set.}
	\label{fig:blocking_set_example}
\end{figure}
Interestingly, minimal blocking sets are powerful both for the design of kernels and for showing kernelization lower bounds.
In fact, under further (in some sense mild) assumptions on $\GG$, it is possible to reduce the number of connected components of $G - M$ to a polynomial in $|M|$ when there is some constant $c$ bounding the size of any minimal blocking set of any graph in $\GG$~\cite{hols2022elimination}.
In contrast, if $\GG$ is closed under disjoint union  
and minimal blocking sets of graphs in $\GG$ can be arbitrarily large, then no polynomial kernel for \vcDistToG{} exists~\cite{hols2022elimination}.
In fact, the concept of minimal blocking sets is so central to the problem that it yields a complete kernelization dichotomy for minor-closed graph classes $\GG$: Bougeret et al.~\cite{bougeretBridgedepthCharacterizesWhich2022} show that minor-closed $\GG$ has bounded minimal blocking sets if and only if $\GG$ has bounded bridgedepth. 
Therefore, there is a polynomial kernel for \vcDistToG{} if and only if $\GG$ has bounded minimal blocking sets.

Motivated by this importance of minimal blocking sets for \textsc{Vertex Cover}, we explore the same (and related) concepts for the family of \fDeletion{} problems.
For a fixed finite family of graphs $\FF$, an instance consists of a graph $G$ and integer $k$, and the question is whether there is a set $S \subseteq V(G)$ of size at most $k$ such that $G - S$ has no graph of $\FF$ as a minor.
When setting $\FF = \{P_2\}$, one obtains \problemname{Vertex Cover}; $\FF = \{K_3\}$ yields \problemname{Feedback Vertex Set}, and there are many further fundamental problems which are captured by this framework; see e.g., Bougeret et al.~\cite{bougeretKernelizationDichotomiesHitting2026STACS}.
Parameterized by the solution size the problem is known to be \textsf{FPT} for all families $\FF$~\cite{DBLP:conf/icalp/SauST20}. 
When $\FF$ contains a planar graph, Fomin et al.~\cite{fominPlanarFdeletionApproximation2012} give a randomized polynomial kernel, which can be derandomized using an algorithm due to Gupta et al.~\cite{guptaLosingTreewidthSeparating2019}, as noted by Bougeret et al.~\cite{bougeretKernelizationDichotomiesHitting2026STACS}.
But it remains a major open question whether the problem admits a polynomial kernel for all families $\FF$.

It is also known that, for all fixed families $\FF$ the problem \fDeletion{} parameterized by the size of a vertex cover admits a polynomial kernel~\cite{DBLP:journals/jcss/FominJP14}. 
For every fixed finite family $\FF$ of connected graphs and every
fixed integer $d \geq 1$, the problem even admits a
polynomial kernel when parameterized by the distance to the class
of graphs of treedepth at most $d$~\cite{jansen2020polynomial}.

When parameterizing by the distance to a minor-closed graph class $\GG$, some complete dichotomies going beyond the one for \textsc{Vertex Cover} are known.
When $\FF = \{K_3\}$ a polynomial kernel exists if and only if $\GG$ has bounded elimination distance\footnote{The elimination distance~\cite{bulianFixedparameterTractableDistances2017,bulianGraphIsomorphismParameterized2016} is a parameter generalizing treedepth.} to the class of forests~\cite{dekkerKernelizationFeedbackVertex2024}.
This result has been generalized to show that, when all graphs in $\FF$ are biconnected on at least three vertices and at least one graph is planar, then \fDeletion{} parameterized by the distance to $\GG$ admits a polynomial kernel if and only if $\GG$ has bounded elimination distance to the class of $\FF$-minor free graphs~\cite{bougeretKernelizationDichotomiesHitting2026STACS}.
The latter result is obtained by combining a new kernel with known lower bounds due to Dekker and Jansen~\cite{dekkerKernelizationFeedbackVertex2024}.
Recently, in an update to the arXiv version of~\cite{bougeretKernelizationDichotomiesHitting2026STACS}, new results have appeared that give a similar characterization also for some non-biconnected families~\cite{bougeret2026kernelizationdichotomieshittingminorsv2}.

The main objects of study in this paper are blocking set notions for \fDeletion{} problems and the relation to \fDeletion{} parameterized by the vertex-deletion distance to some graph class $\mathcal{G}$, specifically to obtain kernelization lower bounds.
An $\FF$-blocking set of a graph $G$ is a set $X \subseteq V(G)$ with the property that no minimum set $S$ such that $G - S$ is $\FF$-minor free contains all of $X$.
Our aim is to investigate the usefulness of this and related notions for kernelization lower bounds.
Related blocking set notions already play an important role in kernelization for \fDeletion{} and related problems~\cite{jansen2020polynomial,bougeretKernelizationDichotomiesHitting2026STACS,dekkerKernelizationFeedbackVertex2024,bougeretKernelizationDichotomiesHittingSubgraphs2024}.
Extending the earlier result for \textsc{Vertex Cover}~\cite{hols2022elimination}, Greilhuber and Sharma~\cite{greilhuber2024component} derive quantitative kernelization lower bounds for $d$-\textsc{Component Order Connectivity} from the size of minimal $\FF$-blocking sets.
However, there is no general result linking large minimal blocking sets to kernelization lower bounds for \fDeletion{} problems.

\subsection{Our Results}

As our first main result, we obtain kernelization lower bounds for \fDeletion{} from large minimal $\FF$-blocking sets.
Given a graph $G$ and integer $c$, the graph $c \cdot G$ is the graph created by taking the disjoint union of $c$ copies of $G$.
The problem \fDeletionDistToG{} is \fDeletion{} parameterized by the size of a set $M$ given as part of the input such that $G - M \in \GG$.

\begin{restatable}{theorem}{mainThmBlockingSetsLB}
    \label{main_thm:blocking_sets_kernel_lb}
    Let $\FF$ be a nonempty finite family of connected graphs on at least two vertices, and $\GG$ a graph class.
    Assume there is some $G \in \GG$ with a minimal $\FF$-blocking set $B$, such that for any integer $c \geq 1$ the graph $c \cdot G$ is in $\GG$.
    Then, \fDeletionDistToG{} has no compression of size $O(|M|^{|B| - \varepsilon})$, for any $\varepsilon > 0$, unless $\np \subseteq \conppoly$.
\end{restatable}

Our theorem generalizes the known quantitative kernelization lower bounds obtained from minimal $\FF$-blocking sets for \textsc{Vertex Cover}~\cite{hols2022elimination} and $d$-\textsc{Component Order Connectivity}~\cite{greilhuber2024component}.
Note that the assumption that $c \cdot G$ is in $\GG$ is always fulfilled when $\GG$ is closed under disjoint union, which is the case for many natural graph classes.

To exhibit the usefulness of \cref{main_thm:blocking_sets_kernel_lb}, let us explain a common approach for showing kernelization lower bounds in the literature.
The known kernelization lower bounds for \textsc{Vertex Cover}~\cite{bougeretBridgedepthCharacterizesWhich2022,hols2022elimination,majumdarPolynomialKernelsVertex2018,cyganHardnessLosingWidth2014}, and also the lower bound for \fDeletion{} parameterized by the distance to graphs of bounded treewidth~\cite{donkersTuringKernelizationDichotomy2021} essentially function this way (for connected $\mathcal{F}$).
The lower bound is typically shown by a reduction from \textsc{CNF-SAT} parameterized by the number of variables (or \textsc{Hitting Set} parameterized by the universe size) for which compression lower bounds are known~\cite{DellM2014}.
For each variable a variable gadget of constant size is created, and, similarly, for each clause a clause gadget is added.
Interestingly, all these reductions share a common core: the central property of each clause gadget is that it has a minimal $\FF$-blocking set $X$, where $|X|$ matches the number of literals of the clause.
\cref{main_thm:blocking_sets_kernel_lb} shows that this ``gold-standard'' approach for \vcDistToG{} extends to all \fDeletion{} problems with connected families $\FF$.

Given \cref{main_thm:blocking_sets_kernel_lb} and the dichotomy for \textsc{Vertex Cover} by Bougeret et al.~\cite{bougeretBridgedepthCharacterizesWhich2022}, it is a natural question whether all kernelization lower bounds for \fDeletion{} problems with minor-closed $\GG$ can be obtained from $\GG$ having unbounded minimal $\FF$-blocking sets.\footnote{When $\GG$ is not minor-closed, the work of Hols et al.~\cite{hols2022elimination} establishes this is not even true for \textsc{Vertex Cover}.}
We show that this is not the case: for $\FF = \{K_3\}$, that is \problemname{Feedback Vertex Set}, there is a minor-closed graph class $\mathcal{G}$ with bounded minimal $\{K_3\}$-blocking sets, such that there is no polynomial kernel for the problem parameterized by the distance to $\GG$.
In fact, gadgets used for the lower bound of Dekker and Jansen~\cite[Theorem 2]{dekkerKernelizationFeedbackVertex2024} give rise to such a graph class.

\begin{restatable}{theorem}{thmBlockingSetsDoNotSuffice}
    \label{thm:fvs_minimal_blocking_sets_do_not_suffice}
    Let $\FF = \{K_3\}$.
    There is a minor-closed graph class $\mathcal{G}$ that has minimal $\FF$-blocking sets of size at most two, such that \fDeletionDistToG{} has no polynomial kernel, unless $\np \subseteq \conppoly$.
\end{restatable}

Thus, the basic definition of minimal $\FF$-blocking sets is, while being useful, not the right notion for all \fDeletionDistToG{} problems, even when $\GG$ is minor-closed.

\subparagraph*{Beyond basic blocking sets.}
To mitigate this situation somewhat, we consider a generalization of minimal $\FF$-blocking sets.
This notion goes back to Dekker and Jansen~\cite[Lemma 1]{dekkerKernelizationFeedbackVertex2024} (and even Jansen and Pieterse~\cite[Lemma 3]{jansen2020polynomial} in a more general form), and was algorithmically exploited in their kernel for \textsc{Feedback Vertex Set}.
We use $\opt_{\FF}{(G)}$ to denote the size of a minimum set that makes $G$ free of an $\mathcal{F}$-minor.

\begin{restatable}[Extended Blocking Sets]{definition}{defExtendedBlockingSets}
    \label{def:extended_blocking_sets}
    Given a family of connected graphs $\FF$ and a graph $G$, an \emph{extended $\FF$-blocking set} of $G$ is a pair $(X,Y)$, where $X, Y \subseteq V(G)$ are two disjoint sets, such that there is no set $S \subseteq V(G)$ such that:
    \begin{enumerate}
        \item $|S| = \opt_{\FF}{(G)}$,
        \item $G - S$ is $\FF$-minor free,
        \item $X \subseteq S$, and
        \item no two vertices of $Y \setminus S$ are in the same connected component of $G - S$.
    \end{enumerate}
    An extended $\FF$-blocking set $(X,Y)$ is \emph{minimal} if, for all $X' \subsetneq X$ the pair $(X',Y)$ is not an extended $\FF$-blocking set, and for all $Y' \subsetneq Y$ the pair $(X,Y')$ is not an extended $\FF$-blocking set.

    A graph class $\GG$ has bounded minimal extended $\FF$-blocking sets if there is some constant $c$ such that for each graph $G \in \GG$ and each minimal extended $\FF$-blocking set $(X,Y)$ of $G$ we have $|X \cup Y| \leq c$. 
    If no such constant $c$ exists, then $\GG$ has unbounded minimal extended $\FF$-blocking sets. 
\end{restatable}

For an \fDeletion{} problem, any vertex of a graph is a ``partial obstruction'' in the sense that it may be extended to a graph that is a forbidden minor when combining it with further edges and vertices.
Similarly, any path might be a partial obstruction that corresponds to an edge of a forbidden minor.
Hence, \cref{def:extended_blocking_sets} is a natural extension of blocking sets which takes not just the simplest partial obstructions, i.e., individual vertices, but also slightly more complicated partial obstructions, in this case paths, into account.
See \cref{fig:extended_blocking_set_example} for a small example illustrating the two roles played by $X$ and $Y$.

\begin{figure}[t]
	\centering
	\begin{subfigure}[t]{0.46\linewidth}
		\centering
		\includegraphics[page=1]{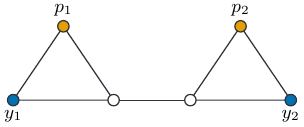}
		\subcaption{The graph with $X = \{p_1,p_2\}$ (orange) and $Y = \{y_1,y_2\}$ (blue).}
	\end{subfigure}
	\hfill
	\begin{subfigure}[t]{0.46\linewidth}
		\centering
		\includegraphics[page=2]{figures/extended-blocking-set-example.pdf}
		\subcaption{Deleting $X = \{p_1,p_2\}$ leaves $y_1,y_2$ in one component.}
	\end{subfigure}
	\caption{A minimal extended blocking set for \textsc{Feedback Vertex Set} ($\FF = \{K_3\}$) on two triangles with a connecting edge ($\opt_\FF = 2$). Selected solution vertices are red, and dashed edges are incident to a deleted vertex.
	For $(X,Y) = (\{p_1,p_2\}, \{y_1,y_2\})$ the only minimum solution containing $X$ is $X$ itself, but then $y_1,y_2$ stay in one component (b), so $(X,Y)$ is an extended $\FF$-blocking set.
	Note $X = \{p_1,p_2\}$ alone is not an $\FF$-blocking set as it is a minimum solution; the additional requirement by $Y$ creates the obstruction.}
	\label{fig:extended_blocking_set_example}
\end{figure}
In fact, even more general notions of minimal blocking sets emerge from the work of Jansen and Pieterse~\cite[Lemma 3]{jansen2020polynomial} that come with great algorithmic power (see also~\cite{bougeretKernelizationDichotomiesHitting2026STACS}), but their usefulness for kernelization lower bounds beyond the notions we study here is not clear (to us).\footnote{In fact, it can be shown that, in general, just from the assumption that such minimal blocking sets are unbounded one cannot obtain kernelization lower bounds. This is essentially a consequence of the fact that this is also already true for minimal extended $\FF$-blocking sets, as we show later.}

Inspired by a reduction of Dekker and Jansen~\cite{dekkerKernelizationFeedbackVertex2024}, we show that a result similar to \cref{main_thm:blocking_sets_kernel_lb} is possible when the graphs in $\FF$ are all biconnected with at least three vertices each.
While complete dichotomies are known for all such families $\FF$ containing a planar graph when $\GG$ is minor-closed~\cite{bougeretKernelizationDichotomiesHitting2026STACS}, our next theorem might still yield new applications in the future, because it does not require $\GG$ to be minor-closed, can also be used to get compression lower bounds for problems that admit polynomial kernels, and does not require $\FF$ to contain a planar graph.

\begin{restatable}{theorem}{thmLBBiconnectedExtendedBlockingSets}
    \label{thm:LB_from_EMBS_biconnected}
    Let $\FF$ be a nonempty family of biconnected graphs on at least three vertices each, and $\GG$ be a graph class.
    Assume that there is some $G \in \GG$ that has a minimal extended $\FF$-blocking set $(X,Y)$ with $|Y| = 2$, such that for any integer $c \geq 1$ the graph $c \cdot G$ is in $\GG$.
    Then, there is no compression for \fDeletionDistToG{} of size $O(|M|^{|X| + 1 - \varepsilon})$, for any $\varepsilon > 0$, unless $\np \subseteq \conppoly$.
\end{restatable}

\cref{thm:LB_from_EMBS_biconnected} follows from a more general theorem that also captures some families that are not biconnected.
For this purpose, we introduce the notion of a \emph{good} minimal extended $\FF$-blocking set.
Following the convention of Donkers and Jansen~\cite{donkersTuringKernelizationDichotomy2021}, a graph is biconnected if it is connected and contains no cut vertex (thus, $K_1$ and $K_2$ are biconnected). 
A leaf block of a graph $F$ is a maximal biconnected subgraph of $F$ that contains at most one cut vertex of $F$.
We denote the set of leaf blocks of graph $F$ by $\leafblock(F)$.

\begin{restatable}[Good Minimal Extended $\FF$-Blocking Sets]{definition}{defGoodExtendedMBS}
    \label{def:good_minimal_extended_blocking_sets}
    Let $\FF$ be a family of connected graphs, and $H$ a graph with a minimal extended $\FF$-blocking set $(X,Y)$.
    Moreover, let $\mathcal{L}$ be the set of leaf blocks of the graphs of $\FF$; that is, $\LL = \bigcup_{F \in \FF} \leafblock(F)$.
    If $Y = \emptyset$, then $(X,Y)$ is a \emph{good} minimal extended $\FF$-blocking set of $H$.

    Otherwise, if $Y \neq \emptyset$, then $(X,Y)$ is a \emph{good} minimal extended $\FF$-blocking set of $H$ if
    each graph of $\LL$ contains at least three vertices, and for each $v \in X$ there exists a set $S_v \subseteq V(H)$ such that:
    \begin{bracketenumerate}
        \item $|S_v| = \opt_{\FF}(H)$,
        \item $H - S_v$ is $\FF$-minor free,
        \item $X \setminus \{v\} \subseteq S_v$,
        \item no two vertices of $Y \setminus S_v$ are in the same connected component of $H - S_v$,
        \item for each $y \in Y \setminus S_v$ the connected component of $H - S_v$ containing $y$ is $\LL$-minor free.
    \end{bracketenumerate}
\end{restatable}

\cref{def:good_minimal_extended_blocking_sets} is inspired by the lower bound of Donkers and Jansen~\cite{donkersTuringKernelizationDichotomy2021}, in which leaf blocks were also central to arguing about the non-existence of forbidden minors.
Observe that for biconnected families $\FF$ containing only graphs with at least three vertices any minimal extended $\FF$-blocking set is automatically good.
Using this notion, we give the following result.

\begin{restatable}{theorem}{thmLBFromExtendedBlockingSets}
    \label{thm:lb_from_extended_blocking_sets}
    Let $\FF$ be a nonempty family of connected graphs on at least two vertices each, and $\GG$ be a graph class.
    Assume that there is some $G \in \GG$ that has a good minimal extended $\FF$-blocking set $(X,Y)$ with $|Y| \in \{0,2\}$, such that for any integer $c \geq 1$ the graph $c \cdot G$ is in $\GG$.
    Then, there is no compression for \fDeletionDistToG{} of size $O(|M|^{|X| + |Y|/2 - \varepsilon})$, for any $\varepsilon > 0$, unless $\np \subseteq \conppoly$.
\end{restatable}

We additionally observe and formally show that the proofs of known kernelization dichotomies~\cite{bougeretBridgedepthCharacterizesWhich2022,bougeretKernelizationDichotomiesHitting2026STACS} yield the following alternative characterization for when \fDeletionDistToG{} admits a polynomial kernel when all graphs in $\FF$ are biconnected and at least one graph of $\FF$ is planar.
Although this result is not difficult to obtain from a careful analysis of the literature, it has not yet been explicitly noted.

\begin{restatable}{theorem}{thmBiconnectedDichotomy}
    \label{main_thm:biconnected_graphs_characterization_blocking_sets}
    Let $\FF$ be a family of biconnected graphs containing at least one planar graph, and $\GG$ a minor-closed graph class. Assume $\np \not \subseteq \conppoly$.
    Then, \fDeletionDistToG{} admits a polynomial kernel if and only if $\GG$ has bounded minimal extended $\FF$-blocking sets.
\end{restatable}

Given \cref{main_thm:biconnected_graphs_characterization_blocking_sets}, one might hope that the notion of minimal extended $\FF$-blocking sets is the central concept for all problem families (and minor-closed $\GG$).
Crushing this hope, with the help of an unpublished result personally communicated to us by Bougeret, Brandwein, and Sau [June 2026], we show that this is not the case.
Concretely, in \cref{thm:k13_extended_blocking_sets_insufficient} we show that the family of caterpillar forests has unbounded minimal extended $\{K_{1,3}\}$-blocking sets. 
However, the upcoming result of Bougeret et al.\ will show that the problem does have a polynomial kernel parameterized by the distance to caterpillar forests.
Consequently, this shows that one cannot simply drop the assumption that the extended $\FF$-blocking set $(X,Y)$ is good from \cref{thm:lb_from_extended_blocking_sets}.

\subparagraph*{New lower bounds.}
We then use our framework to provide new kernelization lower bounds.
The main problem we study is \pwEtaDeletion{}. 
Here the input is a graph $G$ and integer $k$, and the question is whether there is a set $S \subseteq V(G)$ of size at most $k$ such that $G - S$ has pathwidth at most $\eta$.
This is an \fDeletion{} problem obtained by setting $\FF$ to the class of minor-minimal graphs with pathwidth $\eta+1$ (a family $\FF$ is minor-minimal if no graph of $\FF$ is a proper minor of another).
Since the corresponding family $\FF$ contains a planar graph a polynomial kernel exists when parameterizing by the solution-size for each fixed $\eta$~\cite{fominPlanarFdeletionApproximation2012,guptaLosingTreewidthSeparating2019}.
The uniform kernelization complexity of this class of problems was the focus of recent work~\cite{DBLP:conf/esa/GhazyGHS26}.
When $\eta = 0$, \pwEtaDeletion{} is the same as \textsc{Vertex Cover}. 
Kernelization for the special case $\eta = 1$ has also been studied in the literature~\cite{DBLP:conf/wg/PhilipRV10,DBLP:journals/algorithmica/CyganPPW12,DBLP:journals/ipl/Tsur24}.
We consider the parameterization by the distance to $\GG_{\pw \leq \beta}$ where $\GG_{\pw \leq \beta}$ is the class of graphs with pathwidth bounded by some constant $\beta$.
For each $\eta \geq 1$, we show that \pwEtaDeletion{} has no polynomial kernel parameterized by the distance to $\GG_{\pw \leq \eta + 1}$.
We combine this lower bound with known upper bounds~\cite{fominPlanarFdeletionApproximation2012,JansenVertexCoverKernelization2013} and a lower bound for the case when $\eta = 0$~\cite{cyganHardnessLosingWidth2014} to obtain the following kernelization dichotomy.\footnote{Cygan et al.~\cite{cyganHardnessLosingWidth2014} give their lower bound for \textsc{Vertex Cover} parameterized by the distance to the class of graphs with treewidth at most $2$. As observed by Jansen and Pieterse~\cite{jansen2020polynomial}, the proof also works when parameterizing by the distance to the class of graphs with pathwidth at most $2$.}

\begin{restatable}{theorem}{thmPwEtaDichotomy}
    \label{thm:pwEtaDichotomy}
    When $\eta = 0$, \pwEtaDeletion{} admits a polynomial kernel when parameterized by the distance to $\GG_{\pw \leq 1}$, but no polynomial kernel when parameterized by the distance to $\GG_{\pw \leq 2}$, unless $\np \subseteq \conppoly$.

    When $\eta \geq 1$, \pwEtaDeletion{} admits a polynomial kernel when parameterized by the distance to $\GG_{\pw \leq \eta}$, but no polynomial kernel when parameterized by the distance to $\GG_{\pw \leq \eta+1}$, unless $\np \subseteq \conppoly$.
\end{restatable}

A similar result is already known for \twEtaDeletion{} when $\GG$ has bounded treewidth~\cite{cyganHardnessLosingWidth2014,fominPlanarFdeletionApproximation2012,JansenVertexCoverKernelization2013}.
The new results given in \cref{thm:pwEtaDichotomy} are a corollary of a more general and technical result about \fDeletion{} problems parameterized by the distance to graph classes of bounded pathwidth (\cref{thm:pwModulatorBlockingSet}).
We do not mention the formal statement of this result here, but state one more interesting consequence.

\begin{restatable}{theorem}{thmPWMaxPlusOneLB}
    \label{corollary:paramDistPW}
    Let $\FF$ be a nonempty minor-minimal family of connected graphs on at least two vertices each, and set $\eta = \max_{F \in \FF} \pw(F) + 1$.
    Then, \fDeletionDistToG[$\GG_{\pw \leq \eta}$]{} admits no polynomial kernel, unless $\np \subseteq \conppoly$.
\end{restatable}

There are \fDeletion{} problems that do admit a polynomial kernel parameterized by the distance to $\GG_{\pw \leq \eta}$ when $\eta = \max_{F \in \FF} \pw(F)$, concretely the case $\mathcal{F} = \{P_3\}$~\cite{DBLP:conf/mfcs/GreilhuberS26}, and hence we cannot, in general, drop the additive constant of $1$ in the specification of $\eta$.
However, it may be possible to obtain an analogous lower bound when $\eta = \min_{F \in \FF} \pw(F) + 1$.

\subparagraph*{A remark on concurrent related work.}
Largely concurrently with this work, in an update to the arXiv version of~\cite{bougeretKernelizationDichotomiesHitting2026STACS}, Bougeret et al.~\cite{bougeret2026kernelizationdichotomieshittingminorsv2} have generalized the lower bounds of Dekker and Jansen~\cite{dekkerKernelizationFeedbackVertex2024} to also capture some non-biconnected families $\mathcal{F}$.

The proof of these new lower bounds, given in Section 7.2 and Appendix B of~\cite{bougeret2026kernelizationdichotomieshittingminorsv2}, actually also uses gadgets that have unbounded minimal extended $\FF$-blocking sets $(X,Y)$, with $|Y| = 2$, and in general functions similarly to that of Dekker and Jansen~\cite{dekkerKernelizationFeedbackVertex2024}.
In fact, also for those families $\FF$ and graph classes $\GG$ for which the new lower bounds were shown it can be observed that there is a polynomial kernel for \fDeletionDistToG{} if and only if $\GG$ has bounded minimal extended $\FF$-blocking sets.
They additionally provide results (lower and upper bounds) about the size of the kernels in Section~7.3 and Appendix~A. 

However, the assumptions we make in \cref{thm:lb_from_extended_blocking_sets} are not suitable to obtain these new lower bounds as an application from our result.
That is, the minimal extended $\FF$-blocking sets used in the constructions of Bougeret et al.~\cite{bougeret2026kernelizationdichotomieshittingminorsv2} are not necessarily good.
At the same time, the results we give are orthogonal: \cref{thm:lb_from_extended_blocking_sets} can give lower bounds for families $\FF$ that are not in the scope of the new lower-bound results of Bougeret et al.~\cite{bougeret2026kernelizationdichotomieshittingminorsv2}.
For example, if $H$ is the graph consisting of two disjoint triangles connected by an edge, then our lower bound framework (as shown by \cref{corollary:paramDistPW}) can be used to show lower bounds when $\FF = \{H\}$, but this family is out of the scope of the lower-bound results of~\cite{bougeret2026kernelizationdichotomieshittingminorsv2}.

\subsection{Conclusion}

In this work we study the notions of $\FF$-blocking sets and extended $\FF$-blocking sets.
We provide a high-level theorem that shows that, when $\GG$ is closed under disjoint union, $\GG$ having unbounded minimal $\FF$-blocking sets immediately rules out polynomial kernels for \fDeletionDistToG{}, and a similar theorem that applies for unbounded minimal extended $\FF$-blocking sets in some cases.
We hope that these theorems can simplify future lower-bound proofs.

As a new application, we show that, for each integer $\eta \geq 1$, \pwEtaDeletion{} has no polynomial kernel parameterized by the distance to graphs with pathwidth $\eta + 1$.
Together with known algorithmic results and lower bounds, for any $\eta$ and $\beta$ the kernelization complexity of \pwEtaDeletion{} parameterized by the distance to $\GG_{\pw \leq \beta}$ is mapped out.
We additionally make explicit that the kernelization dichotomies for biconnected families $\FF$ depend on $\GG$ having bounded minimal extended $\FF$-blocking sets.

To the best of our knowledge, all currently known dichotomies for \fDeletion{} with connected families $\mathcal{F}$ parameterized by the distance to minor-closed $\GG$ can be shown to depend on the minimal blocking set notions we study.
At the same time, it is highly unclear if the considered notions of minimal blocking sets are the right notions for all \fDeletionDistToG{} problems with connected families $\FF$ and minor-closed $\GG$, or if such a ``catch all'' notion even exists.
Investigating this further is an exciting and promising topic for future work.

The generic lower bounds presented in this paper only work for connected families $\FF$, but at least the lower bound of Donkers and Jansen~\cite{donkersTuringKernelizationDichotomy2021} does not require this assumption.
It seems difficult to get a proof similar to that of \cref{main_thm:blocking_sets_kernel_lb} to work for families $\FF$ that are not connected, but it could be fruitful to investigate this in more detail.
\subsection{Article Organization}

We give a technical overview of our results in \cref{sec:tech_overview}.
\cref{sec:prelims} contains the preliminaries.
\Cref{sec:lower_bounds_from_blocking_sets} introduces our blocking set
notions, proves the lower bound based on good minimal extended
$\FF$-blocking sets (\cref{thm:lb_from_extended_blocking_sets}),
and discusses the necessity of the goodness assumption.
The proof of the dichotomy for biconnected families, \cref{main_thm:biconnected_graphs_characterization_blocking_sets}, is given in \cref{sec:biconnected_dichotomy}.
Finally, \cref{sec:applications} contains the applications of our framework for the parameterization by the distance to graphs of bounded pathwidth, in particular establishing \cref{thm:pwEtaDichotomy,corollary:paramDistPW}.
 
\section{Technical Overview}
\label{sec:tech_overview}
Let us now give a more detailed overview of the used techniques.

\subsection{From Extended Blocking Sets to Lower Bounds}

First, let us sketch the proof of \cref{thm:lb_from_extended_blocking_sets}.
This proof uses gadgets we call universe gadgets, specified as follows.

\begin{restatable}[Universe Gadget]{definition}{defUnivGadget}
    Let $\FF$ be a family of graphs.
    A \emph{universe gadget} is a triple $(H,v,t)$ consisting of a graph $H$, distinguished vertex $v \in V(H)$, and integer $t$, such that:
    \begin{enumerate}
        \item the smallest vertex set $S \subseteq V(H)$ such that $H - S$ is $\FF$-minor free has size $t$,
        \item there is a set $S \subseteq V(H)$ such that $H - S$ is $\FF$-minor free with $|S| = t$, and $v \in S$,
        \item there is a set $S \subseteq V(H)$ such that $H - S$ is $\FF$-minor free with $|S| = t$, and $N(v) \subseteq S$,
        \item and $v$ has exactly two neighbors in $H$.
    \end{enumerate}
\end{restatable}

We show that such a gadget exists for every nonempty family $\mathcal{F}$ of connected graphs on at least two vertices each.
The construction is different depending on whether all graphs in $\FF$ are biconnected: if they are, we adapt a construction of Dekker and Jansen~\cite{dekkerKernelizationFeedbackVertex2024}; otherwise we give a separate construction.
Then, we exploit these gadgets in a linear parameter transformation from $b$-\textsc{Hitting Set} parameterized by the size of the universe $n$, for which no compression of size $O(n^{b-\varepsilon})$ exists for any $\varepsilon > 0$ and any fixed $b \geq 2$, unless $\np \subseteq \conppoly$~\cite{DellM2014}.

\defparproblem{$b$-\problemname{Hitting Set}/$n$}{
    Nonnegative integers $n, s$, and $m$ sets $S_1,\dots,S_m \in \binom{[n]}{b}$.}{
    $n$}{
    Is there a set $S \subseteq [n]$ of size at most $s$
    such that, for every $i \in [m]$, we have $S \cap S_i \neq \emptyset$?}

The reduction we use is similar to (and adapted from) the one used for \textsc{Vertex Cover} by Hols et al.~\cite{hols2022elimination}, for $d$-\textsc{Component Order Connectivity}~\cite{greilhuber2024component}, and a reduction by Dekker and Jansen~\cite{dekkerKernelizationFeedbackVertex2024} for biconnected families $\FF$.
This also makes it similar to further reductions for \textsc{Vertex Cover} such as the ones given in~\cite{fominVertexCoverStructural2016,cyganHardnessLosingWidth2014,majumdarPolynomialKernelsVertex2018}; indeed as mentioned before many reductions exploit unbounded minimal blocking sets similarly, which is one of the main motivations for providing one unified reduction that works for all \fDeletion{} problems with connected $\FF$.

Of course, the proof is adjusted so that it functions correctly with the notion of good minimal extended $\FF$-blocking sets.
In contrast to the reduction of Dekker and Jansen~\cite{dekkerKernelizationFeedbackVertex2024} we also use the set $Y$ of the minimal extended $\FF$-blocking set to encode information if the set is nonempty, and therefore we obtain a slightly larger (and therefore stronger) lower bound than one would get otherwise.

\thmLBFromExtendedBlockingSets*{}

\subsection{A Dichotomy for Biconnected Families}

Now, let us quickly elaborate on \cref{main_thm:biconnected_graphs_characterization_blocking_sets}, the characterization based on the notion of extended $\FF$-blocking sets for when \fDeletionDistToG{} with minor-closed $\GG$ and biconnected graphs $\FF$ admits a polynomial kernel.

\thmBiconnectedDichotomy*{}

If $K_1 \in \FF$, the problem is trivial, and all minor-closed $\GG$ have bounded minimal extended $\FF$-blocking sets.
Otherwise, if $K_2 \in \FF$, the notion of minimal extended $\FF$-blocking sets collapses to the notion of minimal $\FF$-blocking sets, and then this result follows from the work of Bougeret et al.~\cite{bougeretBridgedepthCharacterizesWhich2022}.
It remains to consider the case when all graphs in $\FF$ have at least three vertices.
In this case, we analyze the lower bound by Dekker and Jansen~\cite{dekkerKernelizationFeedbackVertex2024}.
They show that, when $\FF$ contains a planar graph,\footnote{Dekker and Jansen state the result for the case when all graphs in $\FF$ are planar, but as explained in Bougeret et al.~\cite{bougeretKernelizationDichotomiesHitting2026STACS} the proof works as long as there is at least one planar graph in $\FF$.} any minor-closed graph class $\GG$ that has unbounded elimination distance to the class of $\FF$-minor free graphs contains arbitrarily long $\FF$-necklaces.

\begin{restatable}[{$\FF$-Necklace~\cite[Definition 7]{dekkerKernelizationFeedbackVertex2024}}]{definition}{defnecklace}
    \label{def:necklace_graph}
    Let $\FF$ be a nonempty family of connected graphs, and $\ell$ a positive integer.
    An $\FF$-necklace of length $\ell$ is a graph $G$ where $V(G)$ can be partitioned into sets $V_1,\dots,V_\ell$, such that:
    \begin{enumerate}
        \item for each $i \in [\ell]$ the graph $G[V_i]$ is isomorphic to some graph in $\FF$,
        \item there is exactly one edge between $V_i$ and $V_{i+1}$ for each $i \in [\ell-1]$,
        \item there are no edges between $V_i$ and $V_{i'}$ when $i' \notin \{i-1,i,i+1\}$.
    \end{enumerate}
\end{restatable}

We observe that the main property exploited by Dekker and Jansen~\cite{dekkerKernelizationFeedbackVertex2024} for their lower bound is that, for minor-minimal families $\FF$ of biconnected graphs on at least three vertices, the length of an $\FF$-necklace is directly tied to the size of good minimal extended $\FF$-blocking sets, yielding one direction of the characterization.
For the other direction, we show that a recent result of Bougeret et al.~\cite[Lemma 6]{bougeretKernelizationDichotomiesHitting2026STACS}, which bounds the size of certain more complex blocking sets, can be exploited to show that, when $\GG$ has bounded elimination distance to the class of $\FF$-minor free graphs, then $\GG$ also has bounded minimal extended $\FF$-blocking sets.

\subsection{Lower Bounds for Distance to Bounded Pathwidth}

Finally, let us give some details about the new kernelization lower bounds we obtain when parameterizing by the distance to graphs of bounded pathwidth.
For this purpose, we first consider families $\FF$ containing only a single connected graph $F$.
If $F$ is biconnected on at least three vertices, then we observe that for any $\ell \geq 1$ there is a graph with pathwidth $\pw(F)$ that has a good minimal extended $\FF$-blocking set with $|X| = \ell$ and $|Y| = 2$.
In fact, the graphs we use here are again specific $\FF$-necklaces.

\begin{restatable}{lemma}{thmBiconnectedPW}
    \label{thm:biconnectedPW}
    Let $\FF$ be a minor-minimal family of biconnected graphs on at least three vertices.
    Then, for each $F \in \FF$ and for each integer $\ell \geq 1$ there is a graph $G_{\ell}^F$ with pathwidth $\pw(F)$, such that $G_\ell^F$ has a good minimal extended $\FF$-blocking set $(X,Y)$ with $|X| = \ell$ and $|Y| = 2$.
    Moreover, each set $S$ that makes $G_\ell^F$ free of an $F$-minor is such that each biconnected component of $G_\ell^F - S$ has fewer than $|V(F)|$ vertices.
\end{restatable}

If $F$ is not biconnected, and contains at least three vertices, we instead use a slight simplification of a construction of Donkers and Jansen~\cite{donkersTuringKernelizationDichotomy2021}.
Although not made explicit by them, for each integer $\ell \geq 1$ they give a construction yielding a graph with a minimal $\{F\}$-blocking set of size $\ell$.
This construction is easily shown to have the same treewidth as $F$, which (with some considerable additional work) also yielded one main result of~\cite{donkersTuringKernelizationDichotomy2021}.
Luckily for us, due to its linear nature it can also be used for the parameter pathwidth.
We show that, in some specific cases the pathwidth of the construction is the same as $\pw(F)$; in other cases it is at most $\pw(F)+1$.
As mentioned before, one cannot avoid the additive constant of $1$ in this bound for each graph $F$.

Concretely, for our analysis we define the notions of pathwidth criticality and of tidy leaf blocks that we present next.

\begin{restatable}[Pathwidth Critical]{definition}{defPWCritical}
    Given a graph $G$, a vertex $v \in V(G)$ is \emph{pathwidth critical} if $\pw(G) > \pw(G - v)$.
\end{restatable}

\begin{restatable}[Tidy Leaf Blocks]{definition}{defTidyLeafBlock}
    Let $G$ be a connected graph that is not biconnected.
    Then, $G$ has a \emph{tidy leaf block} if there is a smallest (in terms of number of vertices) leaf block $L$ of $G$ such that there is a path decomposition $\XX = X_1,\dots,X_t$ of $L$ of width at most $\pw(G)$, where $X_1 = \{c\}$ and $X_t = \{b\}$, where $c$ is the cut vertex of $G$ contained in $L$, and $b \neq c$ is a vertex of $L$ that is pathwidth critical for $G$.
\end{restatable}

We also exploit a further notion of Donkers and Jansen~\cite{donkersTuringKernelizationDichotomy2021}, which is later needed to obtain the lower bound for families containing more than one graph, the $\beta$-prune of a graph.

\begin{restatable}[Prune {\cite{donkersTuringKernelizationDichotomy2021}}]{definition}{defPrune}
	\label{def:prune}
	Let $\beta$ be an integer and $H$ a graph.
	Then, $\prune{\beta}{H}$ is the possibly empty graph obtained from $H$ by repeatedly choosing a leaf block $L$ of the current graph containing fewer than $\beta$ vertices, and deleting the vertices of $L$ that are not cut vertices of the current graph.
\end{restatable}

We then obtain the following result by analyzing the pathwidth of the construction.

\begin{restatable}{lemma}{thmOneHPW}
    \label{thm:nonBiconnectedPW}
    Let $H$ be a connected graph that is not biconnected.
    If $H$ has a tidy leaf block, set $\omega = \pw(H)$; otherwise set $\omega = \pw(H) + 1$.

    Then, for any integer $\ell \geq 1$ there is a graph $G^\star_\ell$ with pathwidth at most $\omega$ that has a minimal $\{H\}$-blocking set $X$ of size $\ell$.
    Moreover, for any $v \in X$ there is a set $S_v \subseteq V(G^\star_\ell)$ of size $\OPT_{\{H\}}(G^\star_\ell)$ containing $X \setminus \{v\}$, such that $G^\star_\ell - S_v$ is $H$-minor free, and each connected component of $\prune{|V(L)|}{G^\star_\ell - S_v}$ is a minor of $H$, where $L$ is a smallest leaf block of $H$.
\end{restatable}

The remaining case when $F = K_2$ is already handled by the work of Cygan et al.~\cite{cyganHardnessLosingWidth2014}.

To go from families $\FF$ that contain a single graph to minor-minimal families containing multiple graphs, in the spirit of the result of Donkers and Jansen~\cite{donkersTuringKernelizationDichotomy2021}, we exploit the additional guarantees given by these constructions.
There is always a graph $F \in \FF$ such that for this $F$ either the construction for a biconnected graph, or the construction for a non-biconnected graph has a large good minimal extended $\FF$-blocking set.
Concretely, one can safely use the construction based on any graph $F \in \FF$ that has a smallest leaf block among all graphs in $\FF$.
Crucially, for \pwEtaDeletion{} when $\eta \geq 1$, the respective family $\FF$ contains at least one tree on at least three vertices.
Moreover, any tree $T \in \FF$ has a tidy leaf block, and therefore in this case the lower bound we obtain rules out polynomial kernels parameterized by the distance to graphs with pathwidth $\pw(T) = \eta + 1$.
Apart from this result, we obtain the lower bound in \cref{corollary:paramDistPW} from the fact that, since the construction works for some graph $F \in \FF$, it has pathwidth at most $\max_{F \in \FF} \pw(F) + 1$.
 
\section{Preliminaries}
\label{sec:prelims}
For integers $x,z$ we write $\range[x]{z} = \{y \in \mathbb{Z} \mid x \leq y \leq z \}$.
We use $[n]$ as a shorthand for $\range[1]{n}$.
For a set $S$ and integer $b$, we use $\binom{S}{b}$ to denote the set of all subsets of $S$ of size exactly $b$.

\subparagraph*{Graph minors and leaf blocks.}

We use standard graph notation, and will only repeat notions here that are especially relevant for the paper.
All graphs are finite and simple.
For a graph $G$ and integer $c \geq 0$ the graph $c \cdot G$ is the graph obtained by taking the disjoint union of $c$ copies of $G$.
A graph $H$ is a minor of a graph $G$ if $H$ can be obtained from $G$ by repeatedly deleting edges, deleting vertices, or contracting edges.
If $H$ is a minor of $G$, we also say that $G$ has an $H$-minor.
We call a set of graphs $\mathcal{F}$ minor-minimal if there are no two distinct graphs $H_1,H_2 \in \mathcal{F}$ such that $H_1$ is a minor of $H_2$.

Given graphs $G$ and $H$, $G$ is $H$-minor free if $H$ is not a minor of $G$. 
Similarly, for any family of graphs $\FF$, graph $G$ is $\FF$-minor free if no $F \in \FF$ is a minor of $G$.
Note that for any family $\mathcal{F}$ there is a finite family $\mathcal{F}' \subseteq \mathcal{F}$ such that a graph is $\mathcal{F}$-minor free if and only if it is $\mathcal{F}'$-minor free; this is a consequence of the famous graph minor theorem~\cite{DBLP:journals/jct/RobertsonS04}.
Hence, when required for the application of a lemma, we may assume that $\mathcal{F}$ is finite.
Next, we define some notions related to biconnectivity, following the convention of Donkers and Jansen~\cite{donkersTuringKernelizationDichotomy2021}.

\begin{definition}[Biconnectivity and Blocks]
    For a graph $G$ a vertex $v \in V(G)$ is a \emph{cut vertex} if $G - v$ has more connected components than $G$.
    A \emph{bridge} of $G$ is an edge $uv \in E(G)$ such that $G - uv$ has more connected components than $G$.

    A graph is \emph{biconnected} if it is connected and contains no cut vertex.
    A \emph{biconnected component} or \emph{block} of a graph is a maximal biconnected subgraph.
    A block $B$ is a \emph{leaf block} if $B$ contains at most one cut vertex of $G$.
    In particular, a biconnected graph is a leaf block.
\end{definition}

In particular, the used convention gives us that $K_1$ and $K_2$ are biconnected graphs.

Given graphs $H$ and $G$, a \emph{minor-model} of $H$ in $G$ is a function $\varphi: V(H) \to 2^{V(G)}$, such that (1) for distinct $u,v \in V(H)$ the sets $\varphi(v)$ and $\varphi(u)$ are disjoint, (2) for each $v \in V(H)$ the graph $G[\varphi(v)]$ is nonempty and connected, (3) for each edge $uv \in E(H)$ there is an edge from $\varphi(u)$ to $\varphi(v)$ in $G$.
We denote $\varphi(H) = \bigcup_{v \in V(H)} \varphi(v)$.
A minor-model $\varphi$ of $H$ is minimal if there is no minor-model $\varphi'$ of $H$ with $\varphi'(H) \subsetneq \varphi(H)$.
This minor-model notation was also used by Dekker and Jansen~\cite{dekkerKernelizationFeedbackVertex2024}.

The following proposition given by Dekker and Jansen~\cite{dekkerKernelizationFeedbackVertex2024} will also be useful.
\begin{proposition}[{\cite[Proposition 8]{dekkerKernelizationFeedbackVertex2024}}]
    Let $G$ be a graph that has some biconnected graph $H$ as a minor, and let $\varphi$ be a minimal minor-model of $H$.
    Then, $G[\varphi(H)]$ is biconnected, and $G[\varphi(H)]$ is a minor of a biconnected component of $G$.
\end{proposition}

Clearly, this yields the following.
\begin{proposition}
    \label{proposition:biconnected_minors}
    Let $G$ be a graph that has some connected graph $H$ as a minor.
    Then, each biconnected component of $H$ is a minor of a biconnected component of $G$.
\end{proposition}

For a graph family $\FF$ and a graph $G$ we define $\OPT_{\FF}(G)$ to be the size of a smallest set $S \subseteq V(G)$ such that $G - S$ has no graph of $\FF$ as a minor.
When the subscript $\FF$ is clear from the context, we may drop it.

\subparagraph*{Parameterized complexity.}
Let us now define the notions of kernelization and compressions.
We refer to the textbook by Cygan et al.~\cite{cyganParameterizedAlgorithms2015} for more background on the basic notions surrounding parameterized algorithms and kernelization.

\begin{definition}[Parameterized Problems]
    Let $\Sigma$ be a finite alphabet.
    A parameterized problem is a language $L \subseteq \Sigma^* \times \mathbb{N}$. 
    An instance of a parameterized problem is a tuple $(I,k)$ where $I \in \Sigma^*$ and $k \in \mathbb{N}$, where $k$ is called the parameter of the instance.
    The instance is a yes-instance (for problem $L$) if $(I,k) \in L$; otherwise it is a no-instance.
\end{definition}

\begin{definition}[Kernel]
    A kernelization (kernel) of a parameterized problem $L$ is an algorithm that takes instances $(I,k)$ of the problem as input, runs in time polynomial in $|I|+k$, and outputs an instance $(I',k')$ equivalent to $(I,k)$ with $|I'|,k' \leq f(k)$ for some computable function $f$.

    The function $f$ is called the size of the kernel, and if $f$ is a polynomial function, the kernel is a polynomial kernel.
\end{definition}

\begin{definition}[Compression]
    Let $L \subseteq \Sigma^* \times \mathbb{N}$ be a parameterized problem, and $Q \subseteq \Sigma^*$ a language.
    A compression from $L$ into $Q$ is a polynomial-time algorithm that takes instances $(I,k) \in \Sigma^* \times \mathbb{N}$ as input, runs in time polynomial in $|I| + k$, outputs $I' \in \Sigma^*$ such that $I' \in Q$ if and only if $(I,k) \in L$, and $|I'| \leq f(k)$ for some computable function $f$. 

    The function $f$ is called the size of the compression, and if $f$ is a polynomial function, we speak of a polynomial compression.
    When the language $Q$ is not specified, and we speak of the nonexistence of a compression of a certain size, we mean the nonexistence of a compression into any target language $Q$. 
\end{definition}

Clearly, a kernel can be seen as a certain type of compression, and hence ruling out compressions of certain sizes is stronger than ruling out kernels of the same size.

\begin{definition}[Polynomial Parameter Transformation]
    Given two parameterized problems $Q_1$ and $Q_2$, a \emph{polynomial parameter transformation} from $Q_1$ to $Q_2$ is an algorithm that takes instances $(I,k)$ of $Q_1$ as input, runs in polynomial time in $|I| + k$, and outputs an equivalent instance $(I',k')$ of $Q_2$, such that $k' \leq p(k)$ for some polynomial function $p$.
    If $p$ is a linear function, we call the algorithm a \emph{linear parameter transformation}.
\end{definition}

Polynomial parameter transformations can be used to transfer kernelization lower bounds in the natural way.

\subparagraph*{Pathwidth.}

We now define the notion of pathwidth, and 
refer to~\cite[Chapter 7]{cyganParameterizedAlgorithms2015} for more information about this notion and the related notion of treewidth.

\begin{definition}[Pathwidth]
    \label{def:pathwidth}
    A path decomposition of a graph $G$ is a sequence $X_1,\dots,X_t$ of sets, called bags, where $\bigcup_{i \in [t]} X_i = V(G)$, and 
    \begin{enumerate}
        \item for each vertex $v \in V(G)$ there exists an interval $[i,j]$ such that $v$ is in bag $X_\ell$ if and only if $\ell \in [i,j]$,
        \item for each edge $uv \in E(G)$ there is a bag $X_i$ such that $u,v \in X_i$.
    \end{enumerate}
    
    The width of $X_1,\dots,X_t$ is $\max_{i \in [t]} |X_i| - 1$, and the pathwidth of $G$, denoted by $\pw(G)$, is the minimum width of any path decomposition of $G$.
    For any integer $\eta$ we define $\GG_{\pw \leq \eta}$ as the class of all graphs with pathwidth at most $\eta$.
\end{definition}

The definition of treewidth is similar, only that instead of a sequence of bags the bags are arranged according to a decomposition tree.
Since we do not obtain any new results for treewidth we do not formally define the notion here, but we denote the treewidth of a graph $G$ as $\tw(G)$ in some stated results of the literature.

\subparagraph*{Elimination distance.}

The elimination distance is a parameter generalizing treedepth (see, e.g.,~\cite{nesetrilSparsityGraphsStructures2012}), and was introduced by Bulian and Dawar~\cite{bulianFixedparameterTractableDistances2017,bulianGraphIsomorphismParameterized2016}.

\begin{definition}[Elimination Distance]
    \label{def:elim_distance}
    Let $\GG$ be a graph class, $G$ a graph, and $\CC$ the set of connected components of $G$.
    Then, the elimination distance $\upsf{ed}_{\GG}(G)$ of $G$ to $\GG$ is defined inductively as follows:
    \[\upsf{ed}_{\GG}(G) = \begin{cases}
        0 & \text{if $G \in \GG$}, \\
        \max_{C \in \CC} \upsf{ed}_{\GG}(C) & \text{if $G \notin \mathcal{G}$ and $G$ is not connected}, \\
        \min_{v \in V(G)} \upsf{ed}_{\GG}(G - v) + 1 & \text{otherwise.}
    \end{cases}\]
    A graph class $\GG'$ has unbounded elimination distance to $\GG$ if there is no constant $b$ such that all graphs $G \in \GG'$ have $\upsf{ed}_{\GG}(G) \leq b$, and bounded elimination distance to $\GG$ if such a constant $b$ exists.
\end{definition} 
\section{Lower Bounds From Minimal Blocking Sets}
\label{sec:lower_bounds_from_blocking_sets}
We now proceed to our study of kernelization lower bounds via large minimal (extended) $\FF$-blocking sets.
First, in \cref{sec:blocking-sets} we define the most basic notion of $\FF$-blocking sets and explain how they yield lower bounds.
Then, we exhibit the limits of this notion in \cref{sec:basic-blocking-sets-insufficient}.
To mitigate these limitations to some degree, we introduce minimal extended $\FF$-blocking sets in \cref{sec:extended_blocking_sets}.
Finally, we establish that it is not possible that the kernelization dichotomies for \fDeletionDistToG{} for all connected families $\FF$ and minor-closed $\GG$ depend on minimal extended $\FF$-blocking sets being bounded in \cref{sec:on_X_Y_being_good}.

\subsection{Blocking Sets}\label{sec:blocking-sets}

We first define blocking sets in the following most basic sense.

\begin{restatable}[Blocking Set]{definition}{defBlockingSets}
	\label{def:blocking_set}
	Let $\FF$ be a set of graphs, and $G$ a graph.
	A set $X \subseteq V(G)$ is a \emph{$\FF$-blocking set} if all sets $S \subseteq V(G)$ with $X \subseteq S$ where $G - S$ is $\FF$-minor free have size strictly larger than $\OPT_\FF(G)$.
	An $\FF$-blocking set is called \emph{minimal} if it is inclusion-wise minimal with respect to the property of being a blocking set.

	A graph class $\mathcal{G}$ has bounded minimal $\FF$-blocking sets if there is some constant $b$ such that for each graph $G \in \GG$ each minimal $\FF$-blocking set of $G$ has size at most $b$.
	If no such constant $b$ exists, then $\GG$ has unbounded minimal $\FF$-blocking sets.
\end{restatable}

If $\FF$ is clear from the context, we may refer to a (minimal) $\FF$-blocking set simply as a (minimal) blocking set.
The strength of minimal blocking sets comes from the fact that having unbounded minimal blocking sets for $\GG$ in the problem \fDeletionDistToG{} immediately yields a kernelization lower bound (under some further mild conditions).
We prove \cref{main_thm:blocking_sets_kernel_lb} in \cref{sec:extended_blocking_sets} as a corollary of a more general result.

\mainThmBlockingSetsLB*{}

Although there are limitations to \cref{main_thm:blocking_sets_kernel_lb}, the result is already useful for providing lower bounds.
For example, the lower bounds for \vcDistToG{} parameterized by the distance to minor-closed $\mathcal{G}$ given in the literature~\cite{bougeretBridgedepthCharacterizesWhich2022,hols2022elimination,majumdarPolynomialKernelsVertex2018,cyganHardnessLosingWidth2014}
can all be obtained from \cref{main_thm:blocking_sets_kernel_lb} together with appropriate constructions for the graph $G$.
Similarly, the result of Donkers and Jansen~\cite{donkersTuringKernelizationDichotomy2021} for \fDeletion{} parameterized by the distance to graphs of bounded treewidth can be obtained from \cref{main_thm:blocking_sets_kernel_lb} for connected families $\FF$.
We explain this construction in much more detail later in \cref{sec:applications}.

However, as mentioned earlier, there are lower bounds that cannot be obtained from \cref{main_thm:blocking_sets_kernel_lb}.
The next subsection shows this in detail.

\subsection{Having Bounded Minimal Blocking Sets Does Not Suffice}
\label{sec:basic-blocking-sets-insufficient}

We observe that for the problem \textsc{Feedback Vertex Set}, obtained by setting $\FF = \{K_3\}$, there is a \emph{minor-closed} graph class $\GG$ such that \problemname{Feedback Vertex Set} parameterized by the distance to $\GG$ does not admit a polynomial kernel (assuming $\np \not \subseteq \conppoly$), even though minimal $\{K_{3}\}$-blocking sets of $\GG$ are bounded.

We remark that for $\FF = \{K_2\}$, that is for \textsc{Vertex Cover}, the work of Hols et al.~\cite[Theorem 1.2]{hols2022elimination} shows that there is some graph class $\GG$ with minimal $\FF$-blocking sets of size one such that the problem does not have a polynomial kernel parameterized by the distance to $\GG$, unless $\np \subseteq \conppoly$.
However, the graph class they used is not minor-closed, and we remind the reader that for \textsc{Vertex Cover} indeed for any minor-closed graph class $\GG$ with bounded minimal $\FF$-blocking sets there is a polynomial kernel for \textsc{Vertex Cover} parameterized by the distance to $\GG$~\cite{bougeretBridgedepthCharacterizesWhich2022}.

We first introduce the key graphs and graph classes we consider.
\begin{definition}[Triangle Necklace and Triangle Path]
    \label{def:triangle_necklace_triangle_path}
    A \emph{triangle necklace} of length $\ell$ is a graph consisting of $\ell$ pairwise vertex-disjoint triangles $T_1,\dots,T_\ell$, such that each triangle $T_i$ has two distinct vertices $x_i$ and $y_i$, and there is an edge from $y_i$ to $x_{i+1}$ for each $1 \leq i \leq \ell - 1$.
    The graph class consisting of all triangle necklaces and their minors is denoted as $\TT$.
    
    A \emph{triangle path} of length $\ell$ consists of a central path $v_1,\dots,v_{\ell+1}$, and for each $1 \leq i \leq \ell$ there is a vertex $w_i$ adjacent to exactly $v_i$ and $v_{i+1}$.
\end{definition}

It is easy to see that $\TT$ has unbounded elimination distance to forests, and therefore \textsc{Feedback Vertex Set} has no polynomial kernel when parameterizing by the distance to $\TT$~\cite{dekkerKernelizationFeedbackVertex2024}.
In fact, the triangle necklaces we defined are a specific instance of the necklaces Dekker and Jansen~\cite{dekkerKernelizationFeedbackVertex2024} consider in their work, and their lower bound is based on their properties.

Next, we define a certain type of edge set of a graph called an \emph{obstruction-avoider}.
These are edge sets which are intuitively not relevant for the existence of any obstructions.

\begin{definition}[Obstruction-Avoiders]
    Let $\FF$ be a family of graphs, and $G$ be a graph. 
    An $\FF$-\emph{obstruction-avoider} of $G$ is an edge set $X \subseteq E(G)$, such that for any $S \subseteq V(G)$ the graph $G - S$ is $\FF$-minor free if and only if the graph $(G - X) - S$ is $\FF$-minor free.
\end{definition}

Now, we give a somewhat general condition in which a blocking set is not minimal based on this notion of $\FF$-obstruction-avoiders.
It is easy to see that, when $\FF$ only contains connected graphs, any minimal $\FF$-blocking set of a graph $G$ must lie within a single connected component.
This simple observation can be generalized to those components that are obtained after deleting $\FF$-obstruction-avoiders.
The next lemma is a generalization of~\cite[Lemma 13]{greilhuber2024component}.

\begin{lemma}
    \label{thm:blocking_sets_and_obstr_avoiders}
    Let $\FF$ be a family of connected graphs, $G$ be a graph, and $B \subseteq V(G)$ be an $\FF$-blocking set.
    Assume that there exists a set $B' \subseteq B$ such that the graph $G - B'$ has an $\FF$-obstruction-avoider $X$, such that there are two distinct connected components of $(G - B') - X$ that contain vertices of $B$.
    Then, $B$ is not a minimal $\FF$-blocking set.
\end{lemma}
\begin{proof}
    Towards a contradiction, let us assume that the set $B$ fulfilling the conditions of the lemma is a minimal $\FF$-blocking set.
    Let $b_1,b_2 \in B \setminus B'$ be vertices that are in different connected components of $(G - B') - X$; by the statement of the lemma they exist.
    Then, there is some optimal solution $S_1$ for $G$ that contains $B \setminus \{b_2\}$.
    Similarly, there is an optimal solution $S_2$ for $G$ that contains $B \setminus \{b_1\}$.
    Let $C$ be the connected component of $(G - B') - X$ that contains $b_1$.
    We argue that then, $S^\star = (S_2 \setminus V(C)) \cup (S_1 \cap V(C))$ is an optimal solution that contains all of $B$.

    Assume that $G - S^\star$ is not $\FF$-minor free.
    Then, let $S' = S^\star \setminus B'$.
    Clearly, $(G - B') - S'$ is not $\FF$-minor free.
    Then, since $X$ is an $\FF$-obstruction-avoider of $G- B'$, some $H \in \FF$ must be a minor of $(G - B' - X) - S'$ too.
    We have that $C$ is a connected component of $G - B' - X$, and therefore $H$ is a minor of $C - S'$, or a minor of $(G - B' - X - C) - S'$.
    However, we have chosen $S_1$ so that $C - S^\star = C - S'$ is $\FF$-minor free.
    Similarly, we have chosen $S_2$ so that $(G - C) - S^\star$ is $\FF$-minor free, and $(G - B' -X -C) - S'$ is a subgraph of $(G - C) - S^\star$.
    Hence, $H$ cannot be a minor of $G - S^\star$.

    Finally, we need to argue that $S^\star$ is of size $\opt_{\FF}(G)$.
    We first show that $S_1 \cap V(C)$ must be an optimal solution for $C$.
    Note that $S_1$ is an optimal solution for graph $G$ that contains $B'$.
    Hence, $S_1' = S_1 \setminus B'$ is an optimal solution for graph $G - B'$.
    Since $X$ is an $\FF$-obstruction-avoider, we have that $S_1'$ is also an optimal solution for $G - B' - X$.
    It is clear that this means that the restriction of $S_1'$ to a connected component of $G - B' - X$ yields an optimal solution for that connected component.
    In particular, $S_1' \cap V(C)$ is an optimal solution for $C$.
    Since $S_1' \cap V(C) = S_1 \cap V(C)$, we have that $S_1 \cap V(C)$ is an optimal solution for $C$.
    Now, it is clear that $|S_2 \cap V(C)| = |S_1 \cap V(C)|$, and as $S_2$ is also an optimal solution for $G$, we find that $|S^\star| = |S_2| = \opt_{\FF}(G)$.
    Therefore, $S^\star$ is also of minimum size.
\end{proof}

We now show that $\TT$ has bounded minimal blocking sets.
\begin{lemma}
    \label{thm:T_bounded_minimal_blocking_sets}
    For any graph $T \in \TT$ any minimal $\{K_3\}$-blocking set of $T$ has size at most $2$.
\end{lemma}
\begin{proof}
    Towards a contradiction, assume that some $T \in \TT$ has a minimal blocking set $B$ of size at least $3$.
    Moreover, let $T$ be a minor-minimal graph in $\TT$ with this property.
    We first show some structural properties of $T$.
    First, observe that $T$ must be connected: $B$ must lie within a single connected component, and therefore $B$ is a blocking set of a single connected component of $T$.
    If $T$ is not connected, this would contradict the minor-minimality of $T$.

    Now, we argue that $T$ has no bridge.
    \begin{claim}
        $T$ has no bridge.
    \end{claim}
    \begin{claimproof}
        Assume that $T$ has a bridge $e$.
        Observe that, since no cycle contains $e$, a set $S \subseteq V(T)$ is a solution for $T$ if and only if it is a solution for $T - e$.
        Therefore, $B$ is also a minimal blocking set of $T - e$, contradicting that $T$ is minor-minimal.
    \end{claimproof}

    Next, we show that connected graphs of $\TT$ without bridges are triangle paths.
    Although this is not difficult to observe, we provide a detailed proof here for completeness.
    \begin{claim}
        $T$ is a triangle path.
    \end{claim}
    \begin{claimproof}
        First, we show that each connected graph of $\TT$ can be obtained the following way: starting from an empty graph, add a path $v_1,\dots,v_t$.
        Then, for each $i \in [t-1]$ one is allowed, but not obligated, to add a vertex $w_i$ adjacent to only $v_i$ and $v_{i+1}$.
        Finally, one can attach arbitrarily many degree-$1$ vertices to each vertex $v_i$ of the path.

        Let $T' \in \TT$ be a connected graph.
        Hence, there is some triangle necklace $T_1'$ such that $T'$ is a minor of $T_1'$.
        Then, $T'$ can be obtained from a subgraph $\hat T$ of $T_1'$ by using edge contractions only.
        The graph $\hat T$ must be connected since contracting an edge cannot change the number of connected components of a graph.
        Graph $\hat T$ can be obtained by deleting edges and vertices of $T'_1$, such that each obtained intermediate graph remains connected.
        To do this, start with a spanning tree of $\hat T$, and extend it to a spanning tree of $T_1'$.
        Repeatedly delete leaf vertices of the spanning tree that are not part of $\hat T$ until no longer possible.
        Then, $\hat T$ can be obtained from the current graph by edge deletions only.

        Hence, there is some sequence $T'_1,\dots,T'_k = T'$, such that $T'_1$ is a triangle necklace, and for each $i \in [2,k]$ we have that $T'_i$ is a minor of $T'_{i-1}$ obtained by deleting an edge, deleting a vertex, or contracting an edge, and $T_i'$ is connected.
        
        Now, we show that $T'$ has the structure we claim by induction on the length of this sequence.
        
        In the base case $k = 1$ we are done since a triangle necklace clearly can be obtained in the described manner.

        For the induction step, for some $\ell \in [2,k]$ assume that graph $T'_{\ell - 1}$ of the sequence can be obtained in the described manner.
        We show that then this also holds for graph $T'_{\ell}$.
        Let $v_1,\dots,v_t$ be the path defining the structure of $T'_{\ell-1}$.

        If $T'_\ell$ is obtained by deleting a vertex $v$, this vertex cannot be part of $v_2,\dots,v_{t-1}$, since deleting such a vertex would disconnect the graph.
        If the deleted vertex is a pendant of the path, or a vertex $w_i$ forming a triangle with two vertices of the path, then it is clear that $T'_\ell$ also has the desired structure.
        Otherwise, without loss of generality $v = v_1$.
        Now, if $v_1$ has no pendants attached to it, then $T'_\ell$ has the claimed structure by taking $v_2,\dots,v_t$ as the central path.
        If a triangle on $w_1,v_1,v_2$ existed, we instead add a pendant to $v_2$.
        If $v_1$ has a pendant, this either contradicts that $T'_\ell$ is connected, or $T'_{\ell}$ is isomorphic to $P_1$ and the claim is trivial.

        Next, consider that $T'_\ell$ is obtained by deleting an edge.
        If the edge is between a vertex of the path and a pendant, this contradicts that $T'_\ell$ is connected.
        If the edge is between a vertex $w_i$ forming a triangle with $v_i,v_{i+1}$ and one of $v_i,v_{i+1}$, then we can add a pendant instead of the triangle.
        Finally, the edge may be between the path vertices $v_i$ and $v_{i+1}$.
        Then, deleting the edge would disconnect the graph if there is no triangle on $w_i,v_i,v_{i+1}$.
        Hence, such a triangle exists.
        In the creation of $T'_\ell$ we can thus use the path $v_1,\dots,v_i,w_i,v_{i+1},\dots,v_t$.

        Finally, consider that $T'_\ell$ is obtained by contracting an edge.
        If the edge is from a vertex $v_i$ to a pendant, contracting the edge produces a graph isomorphic to the graph obtained by deleting the pendant.
        If the edge is from a vertex $w_i$ forming a triangle with $v_i,v_{i+1}$, then contracting the edge is again essentially the same as deleting vertex $w_i$.
        Otherwise, the edge is between two vertices of the path $v_i,v_{i+1}$.
        Now, observe a triangle on $w_i,v_{i},v_{i+1}$ simply becomes a pendant of the contracted vertex, and pendants attached to either $v_i$ or $v_{i+1}$ become pendants of the contracted vertex too.
        Hence, $T'_\ell$ also has the desired structure.

        Finally, we use this information to argue that $T$ is a triangle path.
        Indeed, since $T \in \TT$ is connected, it follows the structure described above.
        However, there cannot be any pendants attached to the vertices of the central path of $T$, since these trivially form bridges.
        Moreover, if two path vertices $v_i,v_{i+1}$ do not form a triangle with a vertex $w_i$, then also $v_iv_{i+1}$ is a bridge.
        Therefore, $T$ is a triangle path.
    \end{claimproof}

    Now, let $T$ consist of the path $v_1,\dots,v_\ell$ and the vertices $w_1,\dots,w_{\ell - 1}$, where each vertex $w_i$ has an edge to $v_i$ and $v_{i+1}$.
    Consider the order $v_1, w_1, v_2, w_2, \dots, w_{\ell-1}, v_{\ell}$; that is, the vertices of $T$ are arranged by increasing subscript, and $v_i$ is ordered before $w_i$.
    Let $z$ be the vertex of $B$ such that there is some $z' \in B$ coming before $z$ in the order, and some $z'' \in B$ coming after $z$ in the order.
    Such a $z$ exists because $|B| \geq 3$.

    If $z$ is a vertex $v_i$, set $B' = \{v_i\}$ and observe that $T - B'$ has two distinct connected components that contain vertices of $B \setminus B'$, namely a connected component containing $z'$, and a different connected component containing $z''$.
    Applying \cref{thm:blocking_sets_and_obstr_avoiders} with an empty $\{K_3\}$-obstruction-avoider suffices to show that $B$ is not a minimal blocking set.

    Otherwise, $z$ is a vertex $w_i$.
    In this case, set $B' = \{w_i\}$.
    Then, the edge $v_iv_{i+1}$ is a bridge of $T - B'$, and therefore $X = \{v_iv_{i+1}\}$ is a $\{K_3\}$-obstruction-avoider of $T -B'$, and $T- B' - X$ contains two connected components containing vertices of $B \setminus B'$, namely a connected component containing $z'$, and a different connected component containing $z''$. 
    By \cref{thm:blocking_sets_and_obstr_avoiders} the set $B$ is not a minimal blocking set.
\end{proof}

We can now prove \cref{thm:fvs_minimal_blocking_sets_do_not_suffice}.

\thmBlockingSetsDoNotSuffice*{}
\begin{proof}
    By Dekker and Jansen~\cite{dekkerKernelizationFeedbackVertex2024} the problem \fDeletion{} with $\FF = \{K_{3}\}$ parameterized by the distance to $\mathcal{T}$ has no polynomial kernel.
    However, $\mathcal{T}$ is minor-closed and by \cref{thm:T_bounded_minimal_blocking_sets} it has bounded minimal $\FF$-blocking sets.
\end{proof}

\subsection{Lower Bounds via Extended Minimal Blocking Sets}
\label{sec:extended_blocking_sets}

Now, we show that, with a slight generalization of the definition of blocking sets, we can capture those kernelization lower bounds for biconnected families $\FF$
that eluded \cref{main_thm:blocking_sets_kernel_lb}.
The following notion of blocking sets was (somewhat implicitly) defined in the work of Dekker and Jansen~\cite{dekkerKernelizationFeedbackVertex2024} for the \textsc{Feedback Vertex Set} problem.

\defExtendedBlockingSets*{}

We show that
we can also obtain kernelization lower bounds from this notion of blocking sets when some additional conditions apply.

The lower bound always applies if all graphs in $\FF$ are biconnected.
In addition, we consider specific cases of non-biconnected families in which it also applies.
In these non-biconnected cases, we need to make further assumptions about the graph with a large minimal extended $\FF$-blocking set.

\defGoodExtendedMBS*{}

\cref{def:good_minimal_extended_blocking_sets} is inspired by the lower bound of Donkers and Jansen~\cite{donkersTuringKernelizationDichotomy2021}, in which leaf blocks were also central to arguing about the non-existence of minors.
We can observe that \cref{def:good_minimal_extended_blocking_sets} automatically applies to biconnected families $\mathcal{F}$ on at least three vertices, yielding the following observation.

\begin{observation}
	\label{obs:biconnected_families_extended_blocking_sets_good}
	For all families $\FF$ such that each graph of $\FF$ is biconnected on at least three vertices, and any graph $H$ with minimal extended $\FF$-blocking set $(X,Y)$, we have that $(X,Y)$ is good.
\end{observation}

Similarly, it is clear from the definition that for any graph $H$ with minimal $\FF$-blocking set $X$, the pair $(X,\emptyset)$ is a good minimal extended $\FF$-blocking set of $H$.

\begin{observation}
	\label{obs:minimal_blocking_set_good_extended_blocking_sets}
	Let $\FF$ be a family of connected graphs, and $H$ be a graph with a minimal $\FF$-blocking set $X$.
	Then, $(X,\emptyset)$ is a good minimal extended $\FF$-blocking set of $H$.
\end{observation}

Our goal in this section is obtaining kernelization lower bounds from good minimal extended $\FF$-blocking sets.
\thmLBFromExtendedBlockingSets*{}

Before proceeding to the proof of this theorem, from \cref{obs:biconnected_families_extended_blocking_sets_good,obs:minimal_blocking_set_good_extended_blocking_sets,thm:lb_from_extended_blocking_sets} we obtain the following two results.

\mainThmBlockingSetsLB*{}

\thmLBBiconnectedExtendedBlockingSets*{}

\subsubsection*{Proof of \cref{thm:lb_from_extended_blocking_sets}}

Now, we proceed to the proof of \cref{thm:lb_from_extended_blocking_sets}.
We introduce the notion of a universe gadget, which plays a central role.

\defUnivGadget*{}

We call these gadgets universe gadgets because later, when we reduce from a variant of \textsc{Hitting Set}, we use these gadgets to represent the universe of the instance.
Their crucial property is that they have some optimal solution that selects the distinguished vertex $v$, and also have an optimal solution that selects all neighbors of $v$.

\begin{figure}[t]
	\centering
	\begin{subfigure}[t]{0.34\linewidth}
		\centering
		\includegraphics[page=1]{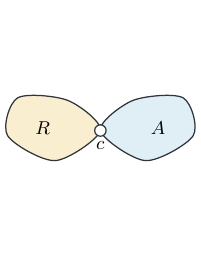}
		\subcaption{The graph $F = R \cup A$.}
	\end{subfigure}
	\hfill
	\begin{subfigure}[t]{0.60\linewidth}
		\centering
		\includegraphics[page=2]{figures/universe-gadget.pdf}
		\subcaption{The graph $H$ of the universe gadget.}
	\end{subfigure}
	\caption{The universe gadget when some graph of $\FF$ is not biconnected. Blue regions are copies of the leaf block $A$ and orange regions are copies of $R$.
	In~(b), $H[V(A_1) \cup V(R_1)]$ and $H[V(A_2) \cup V(R_2)]$ are isomorphic to $F$, whereas this is not necessarily the case for $H[V(A_1) \cup V(R_2)]$.}
	\label{fig:universe_gadget}
\end{figure}

First, we show that universe gadgets exist when some graph in $\FF$ is not biconnected.
\begin{lemma}
	\label{thm:univ_gadget_non_biconnected}
	Let $\FF$ be a minor-minimal family of connected graphs, such that some graph in $\FF$ is not biconnected.
	Then, there is a universe gadget $(H,v,2)$.
\end{lemma}
\begin{proof}
	Let $F \in \FF$ be a non-biconnected graph of $\FF$ with the fewest vertices among all such graphs.
	Let $A$ be a leaf block of $F$, and $c$ be the cut vertex of $F$ contained in $A$.
	Set $R = F - (V(A) \setminus \{c\})$.

	Add two disjoint copies $F_1$ and $F_2$ of $F$. 
	For $i \in \{1,2\}$, let $A_i$, $c_i$, and $R_i$ denote the copies of $A$, $c$, and $R$ in $F_i$.

	Add a vertex $v$ to the graph, and make $v$ adjacent to $c_1$ and to $c_2$.
	Finally, let $w_1 \neq c_1$ be a vertex of $A_1$, and $w_2 \neq c_2$ be a vertex of $R_2$.
	Such vertices exist, because $A_1, R_2$ contain at least two vertices each.
	Identify the vertices $w_1$ and $w_2$, and call the resulting vertex $w$.
	This concludes the description of the resulting graph $H$.
	We now prove that $(H,v,2)$ is the desired universe gadget.
	See \cref{fig:universe_gadget} for an illustration of $F$ and $H$.
	We start with the first property.

	\begin{claim}
		Any set $S \subseteq V(H)$ such that $H - S$ is $\FF$-minor free has size at least $2$.
	\end{claim}
	\begin{claimproof}
		Let $S$ be a minimum-sized set such that $H - S$ is $\FF$-minor free.

		First, consider that $w \in S$.
		Then, it is easy to see that $H[\{v\} \cup V(A_2) \cup V(R_1)]$ has $F$ as a minor.
		Indeed, when contracting the edges from $v$ to $c_1$ and from $v$ to $c_2$ in this graph one directly obtains $F$.
		Therefore, $S$ must select at least one more vertex and $|S| \geq 2$.

		Otherwise, $S$ does not contain $w$.
		Then, notice that $H[V(A_1) \cup V(R_1)]$ and $H[V(A_2) \cup V(R_2)]$ are two graphs isomorphic to $F$ that intersect only in vertex $w$.
		Hence, $S$ must select at least one vertex of $(V(A_1) \cup V(R_1)) \setminus \{w\}$, and at least one vertex of $(V(A_2) \cup V(R_2)) \setminus \{w\}$.
	\end{claimproof}

	We proceed to the second property.
	\begin{claim}
		There is a set $S \subseteq V(H)$ of size $2$ such that $H - S$ is $\FF$-minor free, and $v \in S$.
	\end{claim}
	\begin{claimproof}
		Set $S = \{v,w\}$.
		It is easy to see that the graph $H - S$ can be created by taking two disjoint copies of $F$, and deleting one vertex of each of them.
		Since $\FF$ is minor-minimal and each graph of $\FF$ is connected, it is clear that $H - S$ is $\FF$-minor free.
	\end{claimproof}

	Now, we show the third property.
	\begin{claim}
		There is a set $S \subseteq V(H)$ of size $2$ such that $H - S$ is $\FF$-minor free, and $N(v) \subseteq S$.
	\end{claim}
	\begin{claimproof}
		We set $S = N(v) = \{c_1,c_2\}$.
		In $H - S$ only the connected component of $w$ may not be $\FF$-minor free, as the other connected components are isomorphic to proper subgraphs of $F$.

		Towards a contradiction, assume that some graph $F' \in \FF$ is a minor of $H - S$, and let $\varphi$ be a minimal minor model of $F'$.
		Then, since clearly $H - (S \cup \{w\})$ is $\FF$-minor free, any minor model of $F'$ must use a vertex of $A_1$ and a vertex of $R_2$.
		Moreover, since $A_1, R_2$ are proper subgraphs of a graph of $\FF$, and $\FF$ is minor-minimal, clearly $A_1$ and $R_2$ are $\FF$-minor free.
		Therefore, we must even have that $\varphi$ uses a vertex of $V(A_1) \setminus \{w\}$ and of $V(R_2) \setminus \{w\}$.
		Since $F'$ is connected, and $w$ a cut vertex of $H[V(A_1) \cup V(R_2)]$, there is some $x \in V(F')$ such that $w \in \varphi(x)$.
		Observe that there must be some $x' \in V(F')$ such that $\varphi(x') \subseteq V(A_1)$, and some $x'' \in V(F')$ such that $\varphi(x'') \subseteq V(R_2)$.
		Otherwise $\varphi$ would not be a minimal minor-model.
		Then, since $w$ is a cut vertex of $H[V(A_1) \cup V(R_2)]$, $x$ is a cut vertex of $F'$.

		Hence, the forbidden minor $F'$ is not biconnected.
		However, the model $\varphi$ can use at most $|V(A_1)| + |V(R_2)| - 3$ vertices since $c_1,c_2$ are in $S$, and since both $A_1$ and $R_2$ contain vertex $w$.
		This is a contradiction, since $F$ was chosen to be the smallest non-biconnected forbidden minor, and $|V(F)| = |V(A_1)| + |V(R_2)| - 1$.
	\end{claimproof}
	The lemma follows by combining the claims above and the obvious fact that $v$ has exactly two neighbors in $H$.
\end{proof}

Now, we show that universe gadgets exist when all graphs in $\FF$ are biconnected.
This construction is inspired by the construction of variable gadgets used by Dekker and Jansen~\cite{dekkerKernelizationFeedbackVertex2024}.
We take the same general idea, but need to make some adaptations to account for the concrete specifications of a universe gadget.

	\begin{figure}[t]
		\centering
		\includegraphics{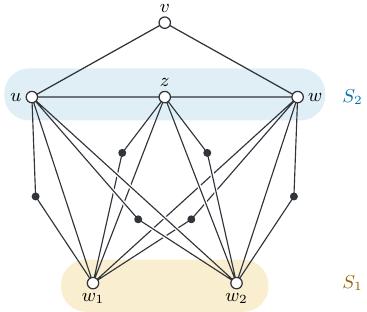}
		\caption{The universe gadget $H$ for $F = K_3$ in the biconnected case of \cref{thm:univ_gadget_biconnected}.
		 The upper blue band marks $S_2 = \{u,z,w\}$, which induces the path $u$--$z$--$w$; the edge $uw$ has been subdivided by $v$. The lower orange band marks $S_1 = \{w_1,w_2\}$.
		 For each of the six pairs $(w_a,s) \in S_1 \times S_2$, a copy of $K_3$ is glued on at $w_a$ and $s$. Its third vertex is a small solid dot.}
		\label{fig:univ_gadget_biconnected}
	\end{figure}

\begin{lemma}
	\label{thm:univ_gadget_biconnected}
	Let $\FF$ be a nonempty minor-minimal family of biconnected graphs on at least two vertices.
	Then, there exists a universe gadget $(H,v,\min_{F \in \FF} |V(F)|)$.
\end{lemma}
\begin{proof}
	Let $F \in \FF$ be one of the graphs of $\FF$ with the fewest vertices.
	Take an edge $uw \in E(F)$, and create $F'$ from $F$ by subdividing this edge.
	All vertices retain their identifier, and the new subdivision vertex is vertex $v$.
	Observe that $v$ has degree $2$, that $F$ is a minor of $F'$, and that $F' - v$ is $\FF$-minor free since $\FF$ is minor-minimal.

	Denote $S_2 = V(F') \setminus \{v\}$.
	Now, add an independent set $S_1$ on $|V(F)|-1$ vertices to the graph.
	For each $w_1 \in S_1$ and each $w_2 \in S_2$ we add a copy of $F$ to the graph, and identify the vertices $w_1$ and $w_2$ with two distinct vertices of $F$.
	Let $H$ be the resulting graph.
	We claim that $(H,v, |V(F)|)$ is a universe gadget.
	See \cref{fig:univ_gadget_biconnected} for a concrete example of the construction.

	First, we show that $\opt_{\FF}(H) \geq |V(F)|$.
	Consider some set $S$ such that $H - S$ is $\FF$-minor free.
	If there is some vertex $w \in S_1$ that $S$ does not contain, then $|S| \geq |V(F)|$, because $w$ is part of $|S_2| = |V(F)|$ copies of $F$ that intersect only in $w$.
	Therefore, let us assume that $S$ contains all vertices of $S_1$.
	Then, if $S$ contains no further vertex, we have $S = S_1$.
	However, $H - S$ is not $\FF$-minor free, because of the $F$-minor on the vertices $S_2 \cup \{v\}$.

	Next, we show that $H' = H - (S_1 \cup \{v\})$ is $\FF$-minor free.
	Since all graphs in $\FF$ are biconnected, by \cref{proposition:biconnected_minors}, it suffices to show that no graph of $\FF$ is a minor of a biconnected component of $H'$.
	First, consider a vertex $w \in V(H') \setminus S_2$.
	Then, any biconnected component containing $w$ has size strictly smaller than $|V(F)|$, because such a vertex $w$ was added as part of a copy of $F$ containing a vertex of $S_1$ and a vertex of $S_2$.
	Hence, no graph of $\FF$ is a minor of a biconnected component that contains a vertex of $V(H') \setminus S_2$.
	Now, consider the biconnected components of $H[S_2]$.
	Observe that $H[S_2]$ is isomorphic to the graph $F - uw$.
	As $\FF$ is minor-minimal, no graph of $\FF$ can be a minor of $H[S_2]$, and therefore also not a minor of a biconnected component containing only vertices of $S_2 $.

	Now, we show that $H - S_2$ is $\FF$-minor free.
	To see this, one can observe that each biconnected component of $H - S_2$ is strictly smaller than $F$.
	Hence, by our choice of $F$ as one of the smallest graphs of $\FF$ the graph $H - S_2$ is $\FF$-minor free.
	This shows the third property of universe gadgets since $N_H(v) \subseteq S_2$.
\end{proof}

By combining \cref{thm:univ_gadget_biconnected,thm:univ_gadget_non_biconnected}, we directly obtain \cref{thm:universe_gadgets_exist}.
The assumption that $\FF$ is minor-minimal can be dropped by simply taking the universe gadget for the family $\FF'$ obtained by making $\FF$ minor-minimal.

\begin{restatable}{lemma}{thmUniverseGadgetExists}
	\label{thm:universe_gadgets_exist}
	Let $\FF$ be a nonempty family of connected graphs on at least two vertices each.
	Then, there is a universe gadget $(H,v,t)$ for some integer $t \geq 1$.
\end{restatable}

For the kernelization lower bound we give in this section, we reduce from $b$-\problemname{Hitting Set} parameterized by the universe size, where $b \geq 1$ is a constant.
\defparproblem{$b$-\problemname{Hitting Set}/$n$}{
	Nonnegative integers $n, s$, and $m$ sets $S_1,\dots,S_m \in \binom{[n]}{b}$.}{
	$n$}{
	Is there a set $S \subseteq [n]$ of size at most $s$
	such that, for every $i \in [m]$, we have $S \cap S_i \neq \emptyset$?}

The problem is useful due to the following compression lower bound.

\begin{lemma}[{Follows from {\cite[Theorem 2]{DellM2014}}}]
	\label{thm:d_hitting_set_lower_bound}
	Let $b\geq 2$.
	Then $b$-\problemname{Hitting Set} parameterized by $n$
	has no compression of size $O(n^{b-\varepsilon})$
	for any $\varepsilon>0$, unless $\np{} \subseteq \conppoly$.
\end{lemma}

Before giving the proof of the main result of this section, we give an intermediate general compression lower bound that we later require for the case when minimal $\FF$-blocking sets have size one.
Hols et al.~\cite{hols2022elimination} also treat this case separately in their kernelization lower bound proof and handle it by repeated application of a hypothetical kernelization algorithm.
However, this iteration does not directly extend to compressions, whose outputs need not be instances of the original problem.
We instead observe that a reduction from $2$-\textsc{Hitting Set} due to Dell and van Melkebeek~\cite[Lemma 7]{DellM2014}, combined with \cref{thm:d_hitting_set_lower_bound}, yields the desired compression lower bound.

\begin{lemma}
	\label{thm:blocking_sets_1_case}
	Let $\FF$ be a nonempty minor-minimal family of connected graphs on at least two vertices each, and $\mathcal{G}$ a nonempty graph class.
	Then, there is no compression for \fDeletionDistToG{} of size $O(|M|^{1 - \varepsilon})$, for any $\varepsilon > 0$, unless $\np \subseteq \conppoly$.
\end{lemma}
\begin{proof}
	Recall that the problem $2$-\textsc{Hitting Set} parameterized by the universe size $n$ does not have a compression of size $O(n^{2 - \varepsilon})$, for any $\varepsilon > 0$, unless $\np \subseteq \conppoly$ (\cref{thm:d_hitting_set_lower_bound}).
Take an input instance of $2$-\textsc{Hitting Set} and discard duplicate sets if there are any.
	The reduction of Dell and van Melkebeek~\cite[Lemma 7 and its proof]{DellM2014}, which refines the reduction in the proof of Lewis and Yannakakis~\cite[Theorem 4]{lewisNodeDeletionProblemHereditary1980}, maps this instance to an equivalent instance of \fDeletion{}.
	Since their construction uses a constant number of copies of the input graph and constant-size gadgets per vertex and edge, the output graph has $O(n^2)$ vertices.
	Therefore, if \fDeletion{} had a compression of size $O(|V(G)|^{1 - \varepsilon})$, for some small $\varepsilon > 0$, then  $2$-\textsc{Hitting Set} would have a compression of size $O(n^{2 - 2\varepsilon})$.
	We conclude that \fDeletion{} does not have a compression of size $O(|V(G)|^{1 - \varepsilon})$, for any $\varepsilon > 0$, unless $\np \subseteq \conppoly$.

	Now, observe that a compression for \fDeletionDistToG{} of size $O(|M|^{1 - \varepsilon})$ for some $\varepsilon > 0$ would, however, give rise to such a compression.
	Indeed, take any instance $(G,k)$ of \fDeletion{}, any $H \in \GG$, let $G^\star$ be the disjoint union of $H$ and $G$, and set $k' = k + \opt{(H)}$.
	Moreover, set $M = V(G)$.
	Then, we have $\opt(G^\star) = \opt(G) + \opt(H)$ since the family $\FF$ only contains connected graphs.
	Therefore, $(G^\star,k',M)$ is an instance of \fDeletionDistToG{} equivalent to the instance $(G,k)$ of \fDeletion{}.
	Hence, there can also not be a compression of size $O(|M|^{1 - \varepsilon})$ for the problem \fDeletionDistToG{}, unless $\np \subseteq \conppoly$.
	This disjoint-union reduction also appears in the work of
	Hols et al.~\cite{hols2022elimination} for \textsc{Vertex Cover}.	
\end{proof}

We also require the following lemma about minors.
\begin{lemma}
	\label{thm:kill_minors_by_leaf_blocks}
	Let $H$ be a connected graph, $G$ be a graph, $c \in V(G)$ be a cut vertex, and $C$ the disjoint union of some connected components of $G - c$ such that no leaf block of $H$ is a minor of $G[V(C) \cup \{c\}]$.
	If $H$ is a minor of $G$, then $H$ is a minor of $G - C$. 
\end{lemma}
\begin{proof}
	Assume the contrary, and let $\varphi$ be a minimal minor model of $H$ in $G$.
	By assumption, $H$ is not a minor of $G - C$.
	Therefore, $\varphi$ contains at least one vertex of $C$.

	First, observe that we must have $c \in \varphi(H)$. 
	Otherwise $H$ is a minor of $C$, and hence some leaf block of $H$ is a minor of $G[V(C) \cup \{c\}]$.
	
	Then, there must be some $v \in V(H)$ such that $\varphi(v) \subseteq V(C)$.
	Otherwise, since $c$ is a cut vertex of $G$, there is only a single vertex $w \in V(H)$ with $\varphi(w) \cap V(C) \neq \emptyset$, and $c \in \varphi(w)$.
	However, this contradicts that $\varphi$ is minimal: if we delete all vertices of $C$ from $\varphi(w)$ to obtain $\varphi'$, then $\varphi'$ is still a minor model of $H$.
	This follows from the fact that no vertex of $C$ has an edge to a vertex of $G - (V(C) \cup \{c\})$, and that $c$ is in $\varphi(w)$.

	Let $W$ be the set of vertices of $V(H)$ such that for each $w \in W$ we have $\varphi(w) \subseteq V(C)$.
	As argued before, $W$ is nonempty.
	Moreover, let $c_H$ be the vertex of $V(H)$ such that $c \in \varphi(c_H)$.
	Since $c$ is a cut vertex, for any $w \in W$ there is not an edge from $\varphi(w)$ to $\varphi(u)$ if $u \notin W \cup \{c_H\}$.
	
	If $V(H) = W \cup \{c_H\}$, then $\varphi(c_H) \cap (V(G) \setminus V(C)) = \{c\}$ because $\varphi$ is a minimal minor model.
	Therefore, $H$ is a minor of $G[V(C) \cup \{c\}]$, contradicting that no leaf block of $H$ is a minor of $G[V(C) \cup \{c\}]$.

	Otherwise, let $\bar W$ be the set of vertices $w$ such that $\varphi(w) \subseteq V(G) \setminus (V(C) \cup \{c\})$, and observe that $\bar W$ is nonempty by the previous paragraph.
	Because $c$ is a cut vertex of $G$, and $C$ a disjoint union of connected components of $G - c$, there is not an edge from a vertex of $\bigcup_{w' \in W} \varphi(w')$ to a vertex of $\varphi(w)$ for any $w \in \bar W$.
	Hence, in $H$ there is not an edge from a vertex of $W$ to a vertex in $\bar W$, and $c_H$ is a cut vertex of $H$.
	Clearly, then $H_W = H[W \cup \{c_H\}]$ is a minor of $G[V(C) \cup \{c\}]$.
	If we show that this means some leaf block of $H$ is a minor of $G[V(C) \cup \{c\}]$ we are done, since this is not the case by assumption.

	So, assume that no leaf block of $H$ is a minor of $G[V(C) \cup \{c\}]$.
	Since $c_H$ is a cut vertex of $H$, there is not a biconnected component of $H$ containing vertices of $W$ and of $\bar W$.
	As $W$ is nonempty, some biconnected components of $H$ contain vertices of $W$.
	Consider the block-cut tree of $H$: the tree with one vertex for each block and cut vertex of $H$, with an edge between a block vertex and a cut vertex if the cut vertex appears in the block.
	Root this tree at vertex $c_H$, and choose some block $L$ that meets $W$ at maximum distance from $c_H$.
	Then, $L$ is a leaf block contained in $H_W$.
	Hence,  $L$ is a minor of $G[V(C) \cup \{c\}]$, contradicting our assumption.
\end{proof}

Now, we are ready to provide the proof of \cref{thm:lb_from_extended_blocking_sets}.
We recall that our proof adapts the lower-bound strategies used by Hols et al.~\cite{hols2022elimination} for \textsc{Vertex Cover}, the proof for $d$-\textsc{Component Order Connectivity}~\cite{greilhuber2024component}, and the proof of Dekker and Jansen~\cite{dekkerKernelizationFeedbackVertex2024}.
In fact, one of the reasons for proving the theorem is to hopefully avoid having to redo essentially the same reduction for each new result in the future.
The main difference is that we use the universe gadgets and the notion of good minimal extended $\FF$-blocking sets here.

\begin{figure}[t]
    \centering
    \begin{subfigure}[t]{0.46\linewidth}
        \centering
        \includegraphics[page=1]{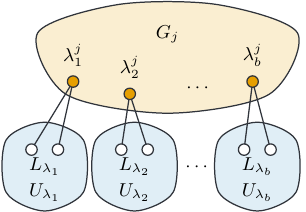}
        \subcaption{The case $|Y_j| = 0$: each $\lambda_\ell^j \in X_j$ is joined to both vertices of $L_{\lambda_\ell}$.}
    \end{subfigure}
    \hfill
    \begin{subfigure}[t]{0.46\linewidth}
        \centering
        \includegraphics[page=2]{figures/reduction-set-gadgets.pdf}
        \subcaption{The case $|Y_j| = 2$: $\lambda_1^j,\dots,\lambda_{b-1}^j$ attach to the gadgets $U_{\lambda_1},\dots,U_{\lambda_{b-1}}$, while $y_1,y_2$ attach to $L_{\lambda_b}$ of $U_{\lambda_b}$.}
    \end{subfigure}
    \caption{How a set gadget $G_j$ is connected to the selection gadgets $U_{\lambda_1},\dots,U_{\lambda_b}$ of $S_j = \{\lambda_1,\dots,\lambda_b\}$, in the two cases of \cref{thm:lb_from_extended_blocking_sets}. Vertices of $X_j$ are orange and vertices of $Y_j$ are blue; a dotted curve represents a possible path within $G_j$.}
    \label{fig:overview_reduction_set_gadgets}
\end{figure}

\begin{figure}[t]
    \centering
    \begin{subfigure}[t]{0.46\linewidth}
        \centering
        \includegraphics[page=1]{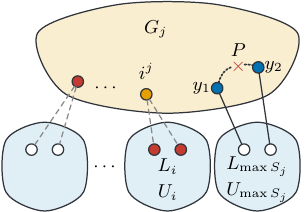}
        \subcaption{$i \neq \max S_j$: only $i^j$ survives in $X_j$.}
    \end{subfigure}
    \hfill
    \begin{subfigure}[t]{0.46\linewidth}
        \centering
        \includegraphics[page=2]{figures/reduction-safety-cases.pdf}
        \subcaption{$i = \max S_j$: all of $X_j$ is deleted.}
    \end{subfigure}
    \caption{The solution $W'$ near $U_i$ and $G_j$ for a chosen $i \in W \cap S_j$, with $|Y| = 2$.
    Red vertices are deleted; dashed edges are incident to them. Orange and blue mark vertices of $X_j$ and $Y_j$ that remain in $G_j - W'$.
    In both panels, $L_i$ is deleted.
    In~(a), the cross indicates that no $y_1$--$y_2$ path survives within $G_j$, either endpoint may also be deleted.
    In~(b), $y_1,y_2$ remain connected within $G_j$.
    Some gadgets and solution vertices are omitted.}
    \label{fig:overview_safety_cases}
\end{figure}

\thmLBFromExtendedBlockingSets*{}
\begin{proof}
	For convenience, set $b = |X| + |Y|/2$.
	Without loss of generality the family $\mathcal{F}$ is minor-minimal.
	If $b \leq 1$, the statement follows from \cref{thm:blocking_sets_1_case}, so we assume $b \geq 2$, and hence $|X| \geq 1$ in the rest of the proof.
	We give a linear parameter transformation from $b$-\textsc{Hitting Set}.

	Let $(n,s,\SSS = \{S_1,\dots,S_m\})$ be the input instance of $b$-\textsc{Hitting Set}.
	If $m = 0$ the instance is a trivial yes-instance, and we output an arbitrary constant-sized yes-instance.
	Otherwise, we begin by adding $n$ copies of the universe gadget $(H,v,t)$, which exists by \cref{thm:universe_gadgets_exist}, to the graph, and denote these as $(H_1,v_1,t),\dots,(H_n,v_n,t)$.
	For each $i \in \range[1]{n}$, let us denote the set $N_{H_i}(v_i)$ as $L_i = \{z_i,z_i'\}$.
	Then, delete the vertex $v_i$ from $H_i$.
	We call $U_i = H_i - \{v_i\}$ the \emph{selection gadget} of element $i$.

	For each $S_j = \{\lambda_1,\dots,\lambda_b\} \in \SSS$, where $\lambda_\ell < \lambda_{\ell+1}$ for all $\ell \in [b-1]$, we add a copy of the graph $G$, and denote it as $G_j$.
	We call the graph $G_j$ the \emph{set gadget} of $S_j$.
	Let $S_j^\star$ be the set $S_j$ if $Y = \emptyset$, and otherwise let it be the set $S_j \setminus \{\lambda_b\}$, that is, all elements of $S_j$ apart from the largest element.
	We denote the good minimal extended $\FF$-blocking set of $G_j$ by $(X_j,Y_j)$, and the vertices of $X_j$ as $\lambda_1^j,\dots,\lambda_{|X|}^j$.
	When $i=\lambda_\ell$, we also write $i^j$ for the vertex $\lambda_\ell^j$ of $X_j$.
	For each $\lambda_\ell \in S_j^\star$ we make the vertex $\lambda_\ell^j$ adjacent to both vertices of $L_{\lambda_\ell}$.
	Moreover, if $|Y_j| = 2$, let $Y_j = \{y_1,y_2\}$.
	We make $y_1$ adjacent to vertex $z_{\lambda_b}$ and $y_2$ adjacent to vertex $z_{\lambda_b}'$.

	This concludes the description of graph $G^\star$ of the output instance.
	See \cref{fig:overview_reduction_set_gadgets} for an illustration of how a set gadget $G_j$ is attached to the selection gadgets of the elements of $S_j$.

	Observe that, for each set $S_j$, and each $i \in S_j^\star$ the graph $G^\star[V(U_i) \cup \{i^j\}]$ is isomorphic to graph $H$.
	Moreover, if $Y_j = \{y_1,y_2\}$, $i = \max S_j$, and $P$ is a path in $G_j$ from $y_1$ to $y_2$, the graph $G^\star[V(U_i) \cup V(P)]$ has $H$ as a minor; indeed contracting $P$ to a single vertex results in a graph isomorphic to $H$.

	We set $k = m \cdot \OPT_\FF(G) + n \cdot (t-1) + s$, and $M = \bigcup_{i \in \range[1]{n}} V(U_i)$.
	Clearly $G^\star - M$ is the disjoint union of $m$ copies of $G$, that is, the graph $m \cdot G$.
	By assumption, we have that $m \cdot G$ is in $\GG$.
	Therefore, $M$ is a modulator to $\GG$.
	Moreover, as the size of each selection gadget is constant we have $|M| \in O(n)$.
	Therefore, to conclude the parameter transformation, it remains to show that the instances $(n,s,\SSS)$ of $b$-\textsc{Hitting Set} and $(G^\star,k,M)$ of \fDeletionDistToG{} are equivalent.

	\begin{claim}
		If $(n,s,\SSS)$  is a yes-instance of $b$-\textsc{Hitting Set}, then $(G^\star,k,M)$ is a yes-instance of \fDeletionDistToG{}.
	\end{claim}
	\begin{claimproof}
		Let $W \subseteq \range[1]{n}$ be a set of size at most $s$ that intersects all sets of $\SSS$.
		Moreover, let $\mathcal{L}$ denote the set of leaf blocks of the graphs of $\FF$; that is, $\LL = \bigcup_{F \in \FF} \leafblock(F)$.
		We build a solution $W'$ for the output instance.

		Recall that, for universe gadget $(H,v,t)$, any vertex set that makes $H$ free of an $\FF$-minor must have size at least $t$; moreover there is some optimal solution $S_1$ selecting $v$, and some optimal solution $S_2$ selecting $N_H(v)$ (which therefore cannot contain $v$).

		For each $i \in W$, in the selection gadget $U_i$ we select the vertices according to the solution $S_2$.
		Therefore, we select $t$ vertices of $U_i$, and in particular we select the vertices $L_i$.
		For each $i \in [n] \setminus W$, in the graph $U_i$, we select vertices according to the solution $S_1$.
		Therefore, we select only $t-1$ vertices of $U_i$ since $U_i$ was created from a copy of $(H,v,t)$ by deleting $v$.

		Call the selection so far $W^\star$. 
		Clearly $|W^\star| \leq n \cdot (t-1) + s$.
		Now, we iterate over all sets $S_j \in \SSS$ and consider the set gadget $G_j$.
		Let $S_j = \{\lambda_1,\dots,\lambda_b\}$, where $\lambda_\ell < \lambda_{\ell+1}$ for all $\ell \in [b-1]$.
		There is some element $i \in W \cap S_j$.
		By our choice of $W^\star$, we know that $L_i \subseteq W^\star$.
		Now, let us consider two cases.

		See \cref{fig:overview_safety_cases} for an illustration of the solution $W'$ restricted to the selection gadget $U_i$ and the set gadget $G_j$, in each of the two cases below.

		First, consider that $|Y| = 0$, or that $i$ is not the largest element of $S_j$.
		Then, $i^j$ is an element of $X_j$, and all neighbors of $i^j$ outside $G_j$, the set $L_i$, are part of $W^\star$.
		Since $(X_j,Y_j)$ is a good minimal extended $\FF$-blocking set,
		there is some set $W_j \subseteq V(G_j)$ of size $\opt_{\FF}(G_j)$ such that $G_j - W_j$ is $\FF$-minor free, $W_j$ contains all vertices of $X_j \setminus \{i^j\}$, the vertices of $Y_j$ are disconnected (or deleted) in $G_j - W_j$, and each vertex of $Y_j$ that is not part of $W_j$ is in a connected component of $G_j - W_j$ that is $\mathcal{L}$-minor free.
		We add $W_j$ to our solution.

		Otherwise, we have $|Y| = 2$ and that $i$ is the largest element of $S_j$.
		Then, let $Y_j = \{y_1,y_2\}$, and observe that all neighbors of $Y_j$ outside $G_j$, the set $L_i$, are part of $W^\star$.
		Since $(X_j,Y_j)$ is a good minimal extended $\FF$-blocking set, there is some set $W_j$ of size $\opt_{\FF}(G_j)$ such that $G_j - W_j$ is $\FF$-minor free, and $W_j$ contains all vertices of $X_j$.
		We add $W_j$ to our solution.

		After doing this for all sets in $\SSS$, we obtain a set $W'$ of size  at most $k$ since, in addition to $W^\star$, we add exactly $\opt_{\FF}(G)$ vertices for each set of the input instance.
		It is clear from our selection that for each graph $G_j$ we have that $G_j - W'$ is $\FF$-minor free, and similarly that for each selection gadget $U_i$ we have that $U_i - W'$ is $\FF$-minor free.
		Therefore, if $G^\star - W'$ had some graph $F \in \FF$ as a minor, any minimal minor-model of $F$ must use vertices of a selection gadget and of a set gadget.

		It is clear that any edges from a selection gadget $U_i$ to a set gadget $G_j$ must be from $U_i$ to a vertex of $X_j \cup Y_j$.
		There cannot be any edges to vertices of $X_j$: each vertex of $X_j$ is either part of $W'$, or all of its neighbors in the selection gadgets are part of $W'$.

		Therefore, the only edges between selection gadgets and set gadgets that may exist in $G^\star - W'$ are from a selection gadget $U_i$ to vertices of a set $Y_j$.
		This can only occur if the smallest leaf block of a graph of $\FF$ contains at least three vertices: indeed, otherwise $(X,Y)$ would not be a \emph{good} minimal extended $\FF$-blocking set.
		Assume that there is some minimal minor-model $\varphi: V(F) \to 2^{V(G^\star - W')}$ of $F \in \FF$ that uses a vertex of $G_j$ and a vertex of $U_i$.
		Then, $\varphi(F)$ must contain a vertex of $Y_j = \{y_1,y_2\}$, and a vertex of $L_i$ as these are the only vertices of $G_j$ and $U_i$ that have edges between them.
		Without loss of generality, assume that $\varphi(F)$ contains vertex $y_1$.

		Let $C$ be the connected component of $G_j - W'$ that contains $y_1$.
		By the properties of the set $W'$, there is no path from $y_1$ to $y_2$ in the graph $G_j - W'$.
		Moreover, each vertex $i^j$ of $X_j$ is either part of $W'$, or all neighbors of $i^j$ outside $G_j$ are part of $W'$.
		Therefore, $y_1$ is a cut vertex of $G^\star[\varphi(F)]$, or $\varphi(F) \cap V(C) = \{y_1\}$.

		Consider the case that $\varphi(F) \cap V(C) \neq \{y_1\}$.
		Then, observe that $y_1$ is a cut vertex of $G^\star[\varphi(F)]$, and by the properties of good minimal extended $\FF$-blocking sets, the graph $C$ is $\FF$-minor free, and even $\LL$-minor free.
		Hence, we can apply \cref{thm:kill_minors_by_leaf_blocks}, which shows that $F$ is a minor of $G^\star[\varphi(F)] - (V(C) \setminus \{y_1\})$.
		However, this contradicts that $\varphi$ is a minimal minor model.

		So, finally consider that $\varphi(F) \cap V(C) = \{y_1\}$.
		Let $x \in V(F)$ such that $y_1 \in \varphi(x)$.
		If $\varphi(x) \neq \{y_1\}$, then deleting $y_1$ from $\varphi(x)$ would still result in a minor-model of $F$, contradicting the minimality of $\varphi$.
		Otherwise, $\varphi(x) = \{y_1\}$, but as $y_1$ has at most one neighbor in $\varphi(F)$, this would imply that $F$ has a vertex of degree one.
		This contradicts that $(X,Y)$ is a \emph{good} minimal extended $\FF$-blocking set, because one of the properties of such a set when $Y \neq \emptyset$ is that each leaf block of each graph of $\FF$ contains at least three vertices.
	\end{claimproof}

	To complete the proof we proceed to the backward direction of safety.
	\begin{claim}
		If $(G^\star,k,M)$ is a yes-instance of \fDeletionDistToG{}, then $(n,s,\SSS)$ is a yes-instance of $b$-\textsc{Hitting Set}.
	\end{claim}
	\begin{claimproof}
		Let $W' \subseteq V(G^\star)$ be a set of size at most $k$ such that $G^\star - W'$ is $\FF$-minor free.
		Let $\SSS'$ be those sets $S_j$ of $\SSS$ where $W'$ selects more than $\opt_{\FF}(G)$ vertices of $G_j$.
		Moreover, let $W$ be the set of elements $i \in \range[1]{n}$ where $W'$ selects at least $t$ vertices of $U_i$.

		We prove that $W$ is a set of size at most $s - |\SSS'|$ that intersects all sets of $\SSS \setminus \SSS'$.
		It is then clear that this means the input instance is a yes-instance, as one can add an arbitrary element of each set of $\SSS'$ to $W$ to obtain a solution of size at most $s$ of the input instance.

		Regarding the size bound, in each selection gadget $U_i$ solution $W'$ must select at least $t-1$ vertices, and in each set gadget $G_j$ solution $W'$ must select at least $\OPT_{\FF}(G)$ vertices.
		Moreover, by definition $W'$ contains more than $\OPT_{\FF}(G)$ vertices of $G_j$ for each $S_j \in \SSS'$.
		Hence, the size of $W$ is indeed bounded by $s - |\SSS'|$.

		Therefore, we now argue that $W$ hits $\SSS \setminus \SSS'$.
		Towards a contradiction, assume there is some $S_j \in \SSS \setminus \SSS'$ such that $W \cap S_j = \emptyset$.
		Let $S_j = \{\lambda_1,\dots,\lambda_b\}$, where $\lambda_\ell < \lambda_{\ell+1}$ for all $\ell \in [b-1]$.
		For each $i \in S_j$, since $i \notin W$, we have that $W'$ selects fewer than $t$ vertices of $U_i$.
		Since, for each $i \in S_j^\star$ the graph $G^\star[V(U_i) \cup \{i^j\}]$ is isomorphic to the universe gadget graph $H$, and $\opt_{\FF}(H) = t$, we know that the vertex $i^j$ must be selected.
		If $Y_j = \{y_1,y_2\}$, then for $i = \max S_j$ we have argued that any path $P$ between $y_1$ and $y_2$ in $G_j - W'$ would make $G^\star[V(U_i) \cup V(P)]$ contain $H$ as a minor.
		Therefore, there cannot be a path between the vertices of $Y_j$ in $G_j - W'$.
		Hence, $W' \cap V(G_j)$ contains all of $X_j$ and disconnects the vertices of $Y_j$.
		But, because $(X_j,Y_j)$ is an extended $\FF$-blocking set, this directly yields $|W' \cap V(G_j)| > \opt_{\FF}(G_j)$.
		This contradicts that $S_j$ is in $\SSS \setminus \SSS'$, and finishes the proof.
	\end{claimproof}
	We have thus obtained a linear parameter transformation.
	The result then follows from \cref{thm:d_hitting_set_lower_bound}.
\end{proof}

In \cref{thm:lb_from_extended_blocking_sets} we consider good minimal extended $\FF$-blocking sets since one of the applications we provide later uses them.
Let us stress that having good minimal extended $\FF$-blocking sets is a sufficient condition for the lower bound to work. 
However, it seems likely that a more complex analysis would reveal more cases beyond this one where the lower bound works.
On the other hand, the next section shows that one must add some type of assumptions on top of the existence of unbounded minimal extended $\FF$-blocking sets to obtain kernelization lower bounds.

\subsection{On the Assumption That \texorpdfstring{$\mathbf{(X,Y)}$}{(X,Y)} Is Good}
\label{sec:on_X_Y_being_good}

One might wonder whether the lower bound given in \cref{thm:lb_from_extended_blocking_sets} can be obtained when dropping the assumption that $(X,Y)$ is a \emph{good} minimal extended $\FF$-blocking set.
That is, can we get a similar theorem when we drop the ``good'' qualification?
While this might seem reasonable at first, it becomes clear that a proof similar to the one used for \cref{thm:lb_from_extended_blocking_sets} should not work without any further assumptions. 
Indeed, having edges between vertices in the universe gadget and the set $Y$ can cause problems for specific families $\FF$.
In fact, using an unpublished result [personal communication with Bougeret, Brandwein, and Sau, June 2026], we can prove that one cannot just drop this assumption under the standard complexity assumptions.

Concretely, upcoming work by Bougeret, Brandwein, and Sau will show that $\{K_{1,3}\}$-\problemname{Minor Deletion} has a polynomial kernel parameterized by the distance to a caterpillar forest.
A caterpillar is a graph consisting of a path to which degree-one vertices are attached by an edge, and a caterpillar forest is a forest in which each connected component is a caterpillar.
Note that caterpillar forests are exactly the graphs with pathwidth at most $1$~\cite[Lemma 2.4]{arnborgMonadicSecondOrder1990}.
We show that the class of caterpillar forests has unbounded minimal extended $\{K_{1,3}\}$-blocking sets, which, together with the aforementioned upcoming result, shows that one cannot in general obtain lower bounds for \fDeletionDistToG{} just from the assumption that $\mathcal{G}$ is minor-closed, closed under disjoint union, and has unbounded minimal extended $\mathcal{F}$-blocking sets.

\begin{figure}[t]
	\centering
	\includegraphics{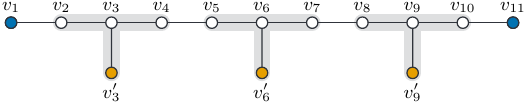}
	\caption{The caterpillar $C_3$. The pendants in $X$ are colored orange, and the path endpoints in $Y$ are colored blue. Gray shading marks three vertex-disjoint claws ($K_{1,3}$) centered at $v_3,v_6,v_9$. For general $n$, the path is extended to $v_1,\dots,v_{3n+2}$, with a pendant attached to every vertex $v_i$ such that $i \equiv_3 0$.}
	\label{fig:caterpillar_C_n}
\end{figure}

\begin{restatable}{theorem}{thmK13ExtendedBlockingSets}
	\label{thm:k13_extended_blocking_sets_insufficient}
	The graph class of caterpillar forests has unbounded minimal extended $\{K_{1,3}\}$-blocking sets.
\end{restatable}
\begin{proof}
	For any integer $n \geq 1$, let $C_n$ be the caterpillar consisting of the path $v_1,\dots,v_{3n +2}$, and for each $v_i$ with $i \equiv_3 0$, that is, for every 3rd vertex of the path, we add a pendant $v_i'$.
	Set $X = \{v_i' \mid i \in [3n+2], i \equiv_3 0\}$ to be the set of all of these additional pendants, and $Y = \{v_1,v_{3n+2}\}$.
	See \cref{fig:caterpillar_C_n} for an illustration of $C_n$ for $n=3$.
	It is easy to see that $C_n$ contains $n$ vertex-disjoint copies of $K_{1,3}$, and therefore one must delete at least $n$ vertices from $C_n$ to make it $K_{1,3}$-minor free.

	Clearly, any set $S \subseteq V(C_n)$ that contains all of $X$ and disconnects all of $Y$ has size larger than $n$, because $|X| = n$ and deleting all of $X$ does not disconnect $Y$.
	Moreover, clearly $X$ is a set such that $C_n - X$ is $\{K_{1,3}\}$-minor free.
	Therefore, $(X,Y)$ is an extended $\{K_{1,3}\}$-blocking set.

	Moreover, $(X,Y)$ is a minimal extended $\{K_{1,3}\}$-blocking set.
	To see this, consider some $X' \subsetneq X$.
	Then, the set $X' \cup \{v_i \mid v_i' \in X \setminus X'\}$ is a solution of size $n$ that contains $X'$ and disconnects $Y$.
	Similarly, $(X,Y')$ with $Y' \subsetneq Y$ is not an extended $\{K_{1,3}\}$-blocking set, because the solution $X$ has optimal size, and as $|Y'| \leq 1$ we vacuously have that no two distinct vertices of $Y'$ remain connected.
\end{proof}
 
\section{The Dichotomy for Biconnected Families and Minimal Extended Blocking Sets}
\label{sec:biconnected_dichotomy}
In this section, we explain that the known dichotomies for biconnected families $\mathcal{F} $ and minor-closed $\GG$ can also be framed in terms of whether $\mathcal{G}$ has unbounded minimal extended $\mathcal{F}$-blocking sets.
That is, we establish \cref{main_thm:biconnected_graphs_characterization_blocking_sets}.

\thmBiconnectedDichotomy*{}

This characterization follows relatively directly from the literature, but since it has not previously been stated explicitly, we consider it worthwhile to record it here.

\subparagraph*{Vertex cover.}
For the case $\FF = \{K_2\}$, that is for \textsc{Vertex Cover}, and minor-closed $\GG$, it is already known that \fDeletionDistToG{} has a polynomial kernel if and only if $\mathcal{G}$ has bounded minimal $\FF$-blocking sets~\cite{bougeretBridgedepthCharacterizesWhich2022}.
From this, we can easily obtain that \fDeletionDistToG{} has a polynomial kernel if and only if $\GG$ has bounded minimal extended $\FF$-blocking sets.
Essentially, the notion of extended $\FF$-blocking sets collapses back to the basic notion of $\FF$-blocking sets for \textsc{Vertex Cover}.

\begin{lemma}
    \label{thm:vc_extended_nonextended_blocking_sets_same}
    Let $\FF = \{K_2\}$, and $\GG$ be a graph class.
    Then, $\GG$ has bounded minimal $\FF$-blocking sets if and only if $\GG$ has bounded minimal extended $\FF$-blocking sets.
\end{lemma}
\begin{proof}
    For the first direction of the proof, observe that a minimal $\FF$-blocking set $X$ immediately gives rise to a good minimal extended $\FF$-blocking set, namely $(X,\emptyset)$.
    Hence, from $\GG$ having unbounded minimal $\FF$-blocking sets, we obtain that $\GG$ has unbounded minimal extended $\FF$-blocking sets.

    Now, assume that $\GG$ has unbounded minimal extended $\FF$-blocking sets, but minimal $\FF$-blocking sets of $\GG$ have size at most $c$.
    Let $(X,Y)$ be a minimal extended $\FF$-blocking set of some graph $G \in \GG$ with $|X \cup Y| > c$.
    We cannot have $Y = \emptyset$, as then $X$ would be a minimal $\FF$-blocking set of size larger than $c$.
    Since $(X,Y)$ is a minimal extended $\FF$-blocking set with $Y \neq \emptyset$, we have that $(X,\emptyset)$ is not an extended $\FF$-blocking set.
    Therefore, there exists some vertex cover $S \subseteq V(G)$ of size $\opt_{\FF}(G)$ that contains all of $X$.
    However, because $\FF = \{K_2\}$, it is clear that $G - S$ has no edges, and therefore all vertices of $Y$ are disconnected.
    This contradicts that $(X,Y)$ is an extended $\FF$-blocking set.
\end{proof}

\subparagraph*{Graphs on at least three vertices.}
Now, let us proceed to the families of graphs on at least three vertices.
For this purpose, we recall a graph family defined by Dekker and Jansen~\cite{dekkerKernelizationFeedbackVertex2024}, which they use as gadgets in a kernelization lower bound.
Although not explicitly mentioned by them, we can easily show that the considered family has unbounded minimal extended $\FF$-blocking sets.

\defnecklace*{}

The crucial property is that such necklaces are guaranteed to appear if minor-closed $\GG$ has unbounded elimination distance to being $\FF$-minor free.

\begin{lemma}[Follows from {\cite[Lemma 7]{dekkerKernelizationFeedbackVertex2024}}]
    \label{thm:$b$-biconnected_graphs_necklaces_exist}
    Let $\FF$ be a family of connected graphs that contains at least one planar graph, and $\GG$ be a minor-closed graph class with unbounded elimination distance to the graph class of $\FF$-minor free graphs.
    Then, $\GG$ contains an $\FF$-necklace of length $\ell$ for all $\ell \geq 1$.
\end{lemma}

The statement in Dekker and Jansen~\cite[Lemma 7]{dekkerKernelizationFeedbackVertex2024} assumes that all graphs in $\FF$ are planar, but as also pointed out in~\cite{bougeret2026kernelizationdichotomieshittingminorsv2}, the proofs actually only require at least one planar graph.
Dekker and Jansen~\cite{dekkerKernelizationFeedbackVertex2024} use a (slightly more controlled) variant of these $\FF$-necklaces to show a kernelization lower bound.
They give a polynomial parameter transformation from \textsc{CNF-SAT} parameterized by the number of variables.
We show that, when minor-minimal $\FF$ is a family of biconnected graphs on at least three vertices, $\FF$-necklaces have large good minimal extended $\FF$-blocking sets, which is in fact the key property that makes the lower bound of Dekker and Jansen work.

\begin{lemma}
    \label{thm:biconnected_graphs_large_extended_blocking_sets}
    Let $\FF$ be a minor-minimal family of biconnected graphs on at least three vertices each.
    Then, for any integer $t \geq 1$, an $\FF$-necklace $G$ of length $t$ has a good minimal extended $\FF$-blocking set $(X,Y)$ with $|X| = t$ and $|Y| = 2$.
\end{lemma}
\begin{proof}
    Let $V_1,\dots,V_t$ be the partition of the vertices of $G$, such that $G[V_i]$ is isomorphic to some graph in $\FF$ for each $i$.
    By the definition of an $\FF$-necklace, for each $i \in [t-1]$ there is exactly one edge from $V_i$ to $V_{i+1}$.
    If $t = 1$ let $r_1$ and $\ell_1$ be two distinct but arbitrarily chosen vertices of $V_1$.
    If $t > 1$,
    for each $i \in [t-1]$ let us denote the vertex of $V_i$ that has an edge to $V_{i+1}$ as $r_i$, and for each $i \in \range[2]{t}$ let us denote the vertex of $V_{i}$ that has an edge to $V_{i-1}$ as $\ell_i$.
    Additionally, let us denote some vertex of $V_1 \setminus \{r_1\}$ as $\ell_1$, and some vertex of $V_t \setminus \{\ell_t\}$ as $r_t$.
    Moreover, since each graph of $\FF$ has at least three vertices, for each $i \in [t]$ there is some vertex $z_i \in V_i \setminus \{r_i,\ell_i\}$.

    Set $X = \{z_i \mid i \in [t]\}$ and $Y = \{\ell_1,r_t\}$.
    We will show that $(X,Y)$ is a minimal extended $\FF$-blocking set.

    First, notice that $\opt_{\FF}(G) \geq t$, because $G$ contains $t$ vertex-disjoint subgraphs, each of them isomorphic to a graph in $\FF$.
    Then, observe that any set $S$ selecting one vertex of each set $V_i$ is such that $G - S$ is $\FF$-minor free.
    Indeed, by the minor-minimality of $\FF$, then $G[V_i] - S$ cannot contain an $\FF$-minor for any $i \in [t]$.
    Moreover, any biconnected component of $G - S$ is either a single edge between sets $V_i$ and $V_{i+1}$, or isomorphic to a proper subgraph of a graph in $\FF$.
    As all graphs in $\FF$ are biconnected on at least three vertices, $G - S$ must be $\FF$-minor free.
    Hence, $\opt_{\FF}(G) = t$.

    Now, we show that $(X,Y)$ is an extended $\FF$-blocking set.
    Observe that any set $S\subseteq V(G)$ of size $\opt_{\FF}(G)$ that contains all of $X$ must be identical to $X$.
    However, $G - X$ has a path from $\ell_1$ to $r_t$.
    Indeed, for all $i \in [t]$ neither vertex $r_i$ nor vertex $\ell_i$ is part of $X$.
    So, since each graph of $\FF$ is biconnected on at least three vertices, for each $i \in [t]$ there is a path from $\ell_i$ to $r_i$ using only vertices of $V_i \setminus \{z_i\}$.
    This means that in $G - S$ there is a path from $\ell_1$ to $r_1$, from $r_1$ to $\ell_2$, from $\ell_2$ to $r_2$, and so forth, until vertex $r_t$ is reached.

    Next, observe that it follows from our previous arguments that $X$ is itself an optimal solution.
    Therefore, each pair $(X,Y')$ where $Y' \subsetneq Y$ is not an extended $\FF$-blocking set.

    Finally, we need to show that for all sets $X' \subsetneq X$ the pair $(X',Y)$ is not an extended $\FF$-blocking set.
    Let $z_i$ be one of the vertices in $X \setminus X'$.
    We set $S = X \setminus \{z_i\} \cup \{\ell_i\}$.
    Since $S$ selects one vertex from each set $V_i$, the graph $G - S$ is $\FF$-minor free.
    Moreover, there is no path from $\ell_1$ to $r_t$ in $G - S$, since $G - \{\ell_i\}$ either disconnects the vertices $\ell_1$ and $r_t$, or $i = 1$ and therefore $S$ contains $\ell_1$.

    Hence, $(X,Y)$ is a minimal extended $\FF$-blocking set, and by \cref{obs:biconnected_families_extended_blocking_sets_good} it is also good.
\end{proof}

This will allow us to establish one direction of \cref{main_thm:biconnected_graphs_characterization_blocking_sets}.
For the second direction, we explain that if $\GG$ has bounded elimination distance to the graph class of $\FF$-minor free graphs, then it has bounded minimal extended $\FF$-blocking sets.
Interestingly, merely having a bound on the size of (extended) minimal $\FF$-blocking sets seems to be difficult to exploit algorithmically for general \fDeletion{} problems.
The known state-of-the-art approaches instead use a different notion (which in fact can also be seen as a type of blocking set), based on the notion of labeled minors~\cite{jansen2020polynomial,bougeretKernelizationDichotomiesHitting2026STACS}.
Let us now define the relevant notions for these approaches.

\begin{definition}[Labeled Minors]
    \label{def:labeled_minor}
    Let $Z$ be a set.
    A $Z$-labeled graph $G$ is a graph where each vertex $v$ is assigned a set of labels $X_v \subseteq Z$.
    A $Z$-labeled graph $H$ is a labeled minor of a labeled graph $G$ if $H$ can be obtained from $G$ by repeatedly applying the following operations:
    \begin{enumerate}
        \item deleting a vertex,
        \item deleting an edge,
        \item contracting an edge $uv$, where the label set of the vertex $z$ obtained by the contraction is $X_z = X_u \cup X_v$,
        \item deleting a label from the label set of a vertex.
    \end{enumerate}
\end{definition}

Given a set $Q$ of $Z$-labeled graphs, and an integer $s$, the set $Q$ is \emph{$s$-saturated} if $Q$ contains the graph on a single vertex with label set $Z'$ for each $Z' \subseteq Z$ of size exactly $s$.
Now, we present the ``main lemma'' used to obtain polynomial kernels in this setting.

\begin{lemma}[{\cite[Lemma 6]{bougeretKernelizationDichotomiesHitting2026STACS}}]
    \label{thm:labeled_blocking_sets_bounded_ed}
    Let $\FF$ be a nonempty finite set of connected graphs.
    Let $Q$ be a set of connected $Y$-labeled graphs that is $\min_{F \in \FF} |V(F)|$-saturated, such that each graph of $Q$ is on at most $\max_{H \in \FF} |E(H)| + 1$ vertices.
    Moreover, let $G$ be a $Y$-labeled graph with elimination distance $\eta$ to the class of $\FF$-minor free graphs.

    Assume that, for each minimum solution $S$ to \fDeletion{} on the unlabeled graph underlying $G$, we have that $G - S$ has a graph of $Q$ as a labeled minor.
    Then, there exists a set $Q^* \subseteq Q$, whose size only depends on $\FF$ and $\eta$, such that for each minimum solution $S$ to \fDeletion{} on the unlabeled graph underlying $G$ we have that $G - S$ has a graph of $Q^*$ as a labeled minor.
\end{lemma}

As already pointed out by Bougeret et al.~\cite{bougeretKernelizationDichotomiesHitting2026STACS}, this set $Q^*$ can also be seen as a type of bounded-size blocking set.
Towards completing the proof of \cref{main_thm:biconnected_graphs_characterization_blocking_sets}, we show the not necessarily surprising result that a bound on the blocking sets in the sense of \cref{thm:labeled_blocking_sets_bounded_ed} yields a bound on the size of minimal extended $\FF$-blocking sets.

\begin{lemma}
    \label{thm:bounded_elim_dist_bounds_extended_blocking_sets}
    Let $\FF$ be a nonempty finite family of connected graphs, and $\GG$ be a minor-closed graph class that has elimination distance at most $\eta$ to being $\FF$-minor free.
    Then, there is a constant $c$ depending only on $\eta, \FF$ such that for any $G \in \GG$ and any minimal extended $\FF$-blocking set $(X,Y)$ of $G$ we have $|X \cup Y| \leq c$.
\end{lemma}
\begin{proof}
    Assume the contrary.
    Let $\alpha \geq 1$ be the constant from \cref{thm:labeled_blocking_sets_bounded_ed} that bounds the size of the set $Q^*$. 
    Note that $\alpha$ is indeed a constant since both $\eta$ and $\FF$ are fixed.
    By assumption, there is some $G' \in \GG$ such that $G'$ has a minimal extended $\FF$-blocking set $(X,Y)$ where $|X \cup Y| > 3\alpha + 1$.

    If $\FF$ contains $K_1$, then it is easy to see that there cannot be a minimal extended $\mathcal{F}$-blocking set of the described size, so we shall assume that each graph of $\FF$ contains at least two vertices.

    We create a labeled graph $G$ from $G'$ by assigning the label set $X_v = \{v\}$ to each $v \in X \cup Y$, and the label set $X_v = \emptyset$ to each $v \in V(G) \setminus (X \cup Y)$.
    Create a family $Q_0$ that contains, for each $v \in X$, a graph on a single vertex labeled $\{v\}$, and for each pair of two distinct vertices $u,v \in Y$ a graph on a single edge, with one vertex labeled $\{u\}$ and the other vertex labeled $\{v\}$.
    
    Let $Q$ be a copy of $Q_0$ at this point.
    Then, for each subset $Z$ of $X \cup Y$ of size $\min_{F \in \FF} |V(F)|$, add a graph on a single vertex with label set $Z$ to $Q$; this ensures that $Q$ is $\min_{F \in \FF} |V(F)|$-saturated.

    Let us observe that the set $Q$ and the graph $G$ now fulfill the conditions of \cref{thm:labeled_blocking_sets_bounded_ed}.
    Indeed, because $(X,Y)$ is an extended $\FF$-blocking set, no optimal solution $S$ for $G'$ selects all of $X$ and disconnects each pair of vertices of $Y$.
    It is clear that in either case, this means that $G - S$ has a labeled minor of $Q_0$, namely either the minor we added for the unselected vertex of $X$, or the minor we added for the non-disconnected pair of vertices of $Y$.
    Hence, from \cref{thm:labeled_blocking_sets_bounded_ed} we obtain a set $Q^* \subseteq Q$ of size at most $\alpha$, such that for any optimal solution $S$ of graph $G'$ we have that $G - S$ has a graph of $Q^*$ as a labeled minor.

    We will conclude the proof by showing that this contradicts that $(X,Y)$ is a minimal extended $\FF$-blocking set.

    First, consider that $|X| > \alpha$.
    Then, there is some $H \in Q_0 \setminus Q^*$ that is a single labeled vertex.
    Let the label set of the vertex of $H$ be $\{v\}$. 
    We have $v \in X$.
    Since $(X,Y)$ is a minimal extended $\FF$-blocking set of $G'$, we have that $(X \setminus \{v\}, Y)$ is not an extended $\FF$-blocking set of $G'$.
    So, there is some set $S \subseteq V(G)$ such that $G - S$ is $\FF$-minor free, $|S| = \opt_{\FF}(G)$, and $S$ contains all of $X \setminus \{v\}$ and disconnects all of $Y$.
    Then $G - S$ has a graph of $Q^*$ as a labeled minor.
    Clearly, $G -S$ cannot have a graph of $Q_0 \setminus \{H\}$ as a labeled minor.
    Therefore, $G - S$ must have some graph of $Q \setminus Q_0$ as a labeled minor.
    Now, any connected labeled minor of $G - S$ can use at most two labels: indeed, the only vertices of $G$ that have labels are those of $X \cup Y$, and the only vertex of $X$ present in $G - S$ is $v$.
    Moreover, the vertices of $Y$ are disconnected in $G - S$.
    Hence, all connected labeled minors of $G-S$ can use an empty label set, or a label set that contains a single vertex of $Y \cup \{v\}$, or a label set that contains vertex $v$ together with a single vertex of $Y$.
    If $\min_{F \in \FF} |V(F)| \geq 3$, we are done since then no graph of $Q \setminus Q_0$, and therefore also no graph of $Q^*$, can be a labeled minor of $G - S$.
    So, let us consider that $\min_{F \in \FF} |V(F)| \leq 2$.
    Then, obviously $G - S$ is an independent set (or the empty graph), and cannot contain any graph of $Q \setminus Q_0$ as a labeled minor either.
    This contradicts the properties of $Q^*$.

    Now, consider that $|X| \leq \alpha$.
    Since $|X \cup Y| > 3\alpha + 1$, this implies $|Y| \geq 2\alpha + 1$.
    The set $Q^*$ can contain at most $\alpha$ labeled graphs, and clearly each labeled graph of $Q^*$ is a single labeled vertex, or a single labeled edge.
    Each labeled edge in $Q^*$ touches $2$ vertices of $Y$.
    Therefore, there exists some vertex $u \in Y$ such that $Q^*$ contains no graph on a single edge that has $u$ in its label set.
    We can use arguments essentially analogous to those in the previous case.
    Observe that $(X,Y \setminus \{u\})$ is not an extended $\FF$-blocking set, and define an optimal solution $S$ that selects all of $X$, and disconnects all of $Y \setminus \{u\}$.
    The only vertices of $G - S$ that have labels are those of $Y$.
    Now, notice that at most one of $Y \setminus \{u\}$ may have a path to $u$ in $G - S$. 
    Otherwise, two vertices of $Y \setminus (\{u\} \cup S)$ would remain connected through $u$ in $G - S$.
    Therefore, we again find that any connected labeled minor of $G - S$ uses a label set containing at most two vertices.
    Hence, if $\min_{F \in \FF} |V(F)| \geq 3$, no graph of $Q^* \setminus Q_0$ is a labeled minor of $G - S$; if $\min_{F \in \FF} |V(F)| \leq 2$ then $G - S$ is an independent set (or the empty graph), and cannot contain a graph of $Q \setminus Q_0$ as a labeled minor either.
    But, by the properties of $Q^*$, $G -S$ must have some graph of $Q^*$ as a labeled minor, and hence a graph of $Q_0$ is a labeled minor of $G - S$.
    Clearly, no graph of $Q_0$ that is a single labeled vertex can be a labeled minor of $G - S$, as all such vertices are part of $X$ and thus of $S$.
    Hence, the labeled minor is a graph on a single edge with two labeled vertices of $Y$.
    However, because $G - S$ disconnects $Y \setminus \{u\}$, this means the labeled minor must have a vertex labeled $\{u\}$.
    This contradicts our choice of $u$.
    Therefore, $G - S$ cannot have a graph of $Q^*$ as a labeled minor, contradicting the properties of $Q^*$ and concluding the proof.
\end{proof}

Now, we are ready to prove \cref{main_thm:biconnected_graphs_characterization_blocking_sets}.
As a quick remark, let us mention the same dichotomy remains valid if boundedness is only tested for minimal extended $\mathcal{F}$-blocking sets $(X,Y)$ with $|Y| \leq 2$.
In fact, the respective dichotomy could be proved without \cref{thm:bounded_elim_dist_bounds_extended_blocking_sets}.
We allow arbitrary sets $Y$ as in the literature precedent of Dekker and Jansen~\cite[Lemma 1]{dekkerKernelizationFeedbackVertex2024}.

\thmBiconnectedDichotomy*{}
\begin{proof}
    We can assume that $\mathcal{F}$ is minor-minimal and finite.

    For $\FF = \{K_1\}$, the statement is correct since in this case making a graph $\FF$-minor free requires us to delete all vertices of the input graph, and thus \fDeletionDistToG{} can be solved in polynomial time.
    Since $\GG$ is minor-closed, we can also check if $M$ is a modulator to $\GG$ in polynomial time by using an appropriate minor-testing algorithm~\cite{DBLP:journals/jct/RobertsonS95b}.
    Hence, there is a polynomial kernel.
    Similarly, $\GG$ cannot have unbounded minimal extended $\FF$-blocking sets since all graphs have a single unique set that makes the graph $\FF$-minor free.
    
    Next, consider that $\FF = \{K_2\}$.
    In this case, Bougeret et al.~\cite[Theorems 1.1 and 1.2]{bougeretBridgedepthCharacterizesWhich2022} explicitly show that \fDeletionDistToG{} has a polynomial kernel if and only if $\GG$ has bounded minimal $\FF$-blocking sets.
    By \cref{thm:vc_extended_nonextended_blocking_sets_same} $\GG$ has bounded minimal $\FF$-blocking sets if and only if $\GG$ has bounded minimal extended $\FF$-blocking sets, concluding the proof in this case.

    Now, any family $\FF$ that contains $K_1$ or $K_2$ collapses to one of the two previous cases when making it minor-minimal.
    Hence, for the rest of the proof assume that each graph of $\FF$ has at least three vertices.
    First, consider that $\GG$ has bounded elimination distance to the class of $\FF$-minor free graphs.
    Then, by~\cite[Theorem 1]{bougeretKernelizationDichotomiesHitting2026STACS}, and the fact that \fDeletion{} admits a polynomial kernel parameterized by the solution size when $\FF$ contains a planar graph~\cite{fominPlanarFdeletionApproximation2012,guptaLosingTreewidthSeparating2019},
    there is a polynomial kernel for \fDeletionDistToG{}.
    Moreover, by \cref{thm:bounded_elim_dist_bounds_extended_blocking_sets}, $\GG$ has bounded minimal extended $\FF$-blocking sets.

    Now, consider that $\GG$ has unbounded elimination distance to the class of $\FF$-minor free graphs.
    Then, by \cref{thm:$b$-biconnected_graphs_necklaces_exist}, $\GG$ contains an $\FF$-necklace $N_\ell$ of length $\ell$ for all $\ell \geq 1$.
    By \cref{thm:biconnected_graphs_large_extended_blocking_sets} the necklace $N_\ell$ has a good minimal extended $\FF$-blocking set $(X,Y)$ with $|X| = \ell$ and $|Y| = 2$.
    Therefore, $\GG$ has unbounded minimal extended $\FF$-blocking sets.
    Moreover, there is no polynomial kernel by the arguments of~\cite[Theorem 2]{dekkerKernelizationFeedbackVertex2024}.
\end{proof}
 
\section{Parameterizing by the Distance to Bounded-Pathwidth Graphs}
\label{sec:applications}
In this section, we exhibit new kernelization lower bounds that can be obtained using our framework.
The parameter of interest is the deletion distance to graphs of bounded pathwidth.
Recall that the graph class $\GG_{\pw \leq \eta}$ contains exactly the graphs with pathwidth at most $\eta$.

Although our overall result will be slightly stronger, let us mention two important corollaries of it here first.
Recall that \pwEtaDeletion{} is \fDeletion{} where $\FF$ consists of all minor-minimal graphs with pathwidth $\eta+1$.

\begin{restatable}{corollary}{thmPWEtaEtaPlusOne}
	\label{corollary:pwEtaEtaPlusOne}
	Let $\eta \geq 1$ be an integer.
	Then, \pwEtaDeletionEtaPlusOne{} admits no polynomial kernel, unless $\np \subseteq \conppoly$.
\end{restatable}

Note that \pwEtaDeletion{} when $\eta = 0$ is the same as \textsc{Vertex Cover}, and it is known that \textsc{Vertex Cover} admits no polynomial kernel when parameterized by the distance to $\GG_{\pw \leq 2}$~\cite{cyganHardnessLosingWidth2014}, unless $\np \subseteq \conppoly$.
On the other hand, when parameterizing by the distance to $\GG_{\pw \leq 1}$, \textsc{Vertex Cover} does admit a polynomial kernel~\cite{JansenVertexCoverKernelization2013}.
Moreover, \pwEtaDeletion{} parameterized by the distance to $\GG_{\pw \leq \eta}$ admits a polynomial kernel~\cite{fominPlanarFdeletionApproximation2012} since the corresponding family $\FF$ contains a planar graph, and $\GG_{\pw \leq \eta}$ is $\FF$-minor free.\footnote{Since the input contains a modulator to $\GG_{\pw \leq \eta}$, the randomized step in the kernel of Fomin et al.~\cite{fominPlanarFdeletionApproximation2012} can be omitted. If no modulator is supplied, a constant-factor approximate modulator can be computed deterministically in polynomial time using the algorithm of Gupta et al.~\cite{guptaLosingTreewidthSeparating2019}.}
Hence, \cref{corollary:pwEtaEtaPlusOne} together with these results from the literature yields a tight dichotomy for the problem, which we state in \cref{thm:pwEtaDichotomy}.

\thmPwEtaDichotomy*{}

Interestingly, exactly the same dichotomy where pathwidth is replaced with treewidth for the \twEtaDeletion{} problem has been previously established when combining the lower bound of Cygan et al.~\cite{cyganHardnessLosingWidth2014} with the known upper bounds~\cite{JansenVertexCoverKernelization2013,fominPlanarFdeletionApproximation2012}.
Moreover, we obtain the following lower bound result for all connected families $\FF$ containing graphs with at least two vertices each.

\thmPWMaxPlusOneLB*{}

Also recall that, in general, one cannot strengthen the result in \cref{corollary:paramDistPW} to a lower bound parameterized by the distance to graphs with pathwidth at most $\max_{F \in \FF} \pw(F)$, at least when assuming $\np \not\subseteq \conppoly$.
This follows from the work of Greilhuber and Sharma~\cite{DBLP:conf/mfcs/GreilhuberS26}.
Contrasting this, it may be possible to improve the lower bound to rule out polynomial kernels parameterized by the distance to graphs with pathwidth at most $\min_{F \in \FF} \pw(F) + 1$.

For our proof, we begin by considering the case that $\FF = \{H\}$ for a single non-biconnected graph $H$.
In this case, we can show that a simplification of a construction due to Donkers and Jansen~\cite{donkersTuringKernelizationDichotomy2021} yields a lower bound.
Our contribution here is a detailed analysis of the pathwidth of the provided construction.
Then, we consider biconnected families $\FF$ in \cref{sec:biconnected_families}.
In this case, we use a simple observation about necklaces.
Later, we combine these results to obtain lower bounds for connected families $\FF$ in \cref{sec:pw_lb_families}.

We expect that it should also be possible to obtain \cref{corollary:paramDistPW} when $\mathcal{F}$ contains graphs that are not connected.
Donkers and Jansen~\cite{donkersTuringKernelizationDichotomy2021} are able to handle disconnected families $\FF$ in which every graph of $\FF$ contains a connected component on at least three vertices.
However, this requires some additional nontrivial arguments, and as we could not use \cref{thm:lb_from_extended_blocking_sets} for this scenario, the proof would not decompose as nicely.
Therefore, we stick to connected families in this work.

\subsection{Families Containing a Single Non-Biconnected Graph}
\label{sec:non_biconnected_families}

Let us first consider the case that $\FF = \{H\}$ for a single graph $H$ that is not biconnected.
In this case, we can show that there is no polynomial kernel parameterized by the distance to graphs with pathwidth $\pw(H) + 1$.
In fact, in some cases there is not even a polynomial kernel parameterized by the distance to graphs with pathwidth $\pw(H)$.

To show this, we provide a minor adaptation of a construction due to Donkers and Jansen~\cite{donkersTuringKernelizationDichotomy2021} and give a tight analysis of the pathwidth of the construction.
First, let us restate their result.

\begin{theorem}[See {\cite[Theorem 1]{donkersTuringKernelizationDichotomy2021}}]
	Let $\FF$ be a finite set of graphs, such that each graph in $\FF$ has a connected component on at least three vertices.
	Let $\mathcal{G}$ be the class of graphs with treewidth at most $\min_{F \in \FF} \tw(F)$.
	Then, \fDeletionDistToG{} admits no polynomial kernel, unless $\np \subseteq \conppoly$.
\end{theorem}

Although Donkers and Jansen do not consider the distance to graphs with bounded pathwidth as their parameter, the linear nature of their construction makes it a prime candidate for also showing lower bounds for this parameter.
First, we introduce some definitions we will need for the analysis of their construction.

\defPWCritical*{}

\defTidyLeafBlock*{}

Then, we need a further notion of Donkers and Jansen~\cite{donkersTuringKernelizationDichotomy2021}, which is later needed to obtain the lower bound for families containing more than one graph.

\defPrune*{}

We also need the following elementary lemma.
\begin{lemma}
	\label{thm:pathwidth_starting_ending_decomp}
	Let $G$ be a graph, and $u,v \in V(G)$ two (not necessarily distinct) vertices of $G$.
	Then, there is a path decomposition $\mathcal{X} = X_1,\dots,X_t$ of $G$, such that $X_1 = \{u\}$ and $X_t = \{v\}$, and the width of $\mathcal{X}$ is at most $\pw(G) + 1$. 
\end{lemma}
\begin{proof}
	Consider a path decomposition $Y_1,\dots,Y_{t'}$ of $G$ of width $\pw(G)$.
	Let $i$ be the lowest integer such that $u \in Y_i$, and $j$ be the lowest integer such that $v \in Y_j$.
	
	Without loss of generality, we have $i \leq j$.
	Indeed, if this is not the case we can swap $u$ and $v$, and reverse the obtained path decomposition at the end.
	Add vertex $u$ to each bag $Y_\ell$ with $\ell \in [i]$, and add $v$ to each bag $Y_\ell$ with $\ell \in \range[j+1]{t'}$.
	Clearly, we have added at most one new vertex to each bag, and now bag $Y_1$ contains $u$, and bag $Y_{t'}$ contains $v$.
	We can then build the desired decomposition $\mathcal{X} = \{u\},Y_1,\dots,Y_{t'},\{v\}$.
\end{proof}

Now, we are ready to prove the main result of this subsection.

\thmOneHPW*{}
\begin{proof}
	We first repeat the construction of Donkers and Jansen~\cite{donkersTuringKernelizationDichotomy2021} with a slight modification that we elaborate on.
	Apart from this modification, we add an analysis of the pathwidth of the resulting graph to the proof.

	\proofsubparagraph{Graph $\mathbf{H}$.}

	If one of the smallest leaf blocks of $H$ is a tidy leaf block, then let $L$ be this leaf block, $c$ the cut vertex of $H$ contained in leaf block $L$, and $b \in V(L) \setminus \{c\}$ be the corresponding pathwidth critical vertex of $H$.
	Otherwise, let $L$ be an arbitrary leaf block of $H$ with the fewest vertices, $c$ the cut vertex contained in $L$, and $b$ an arbitrary vertex of $V(L) \setminus \{c\}$.
	In either case, let $a$ be some vertex of $H - L$.
	Such a vertex exists because $H$ is not biconnected.
	Let us also denote $R = H - (V(L) \setminus \{c\})$.
	See \cref{fig:H_decomposition} for an illustration.

	\begin{figure}[t]
		\centering
		\begin{subfigure}[b]{0.18\linewidth}
			\centering
			\includegraphics[page=1]{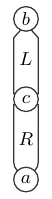}
			\subcaption{The graph $H$.}
			\label{fig:H_decomposition}
		\end{subfigure}
		\hfill
		\begin{subfigure}[b]{0.34\linewidth}
			\centering
			\includegraphics[page=2]{figures/blocks.pdf}
			\subcaption{Gadget $B_1$.}
			\label{fig:block_B1}
		\end{subfigure}
		\hfill
		\begin{subfigure}[b]{0.36\linewidth}
			\centering
			\includegraphics[page=3]{figures/blocks.pdf}
			\subcaption{Gadget $B_2$.}
			\label{fig:block_B2}
		\end{subfigure}
		\caption{The graph $H = R \cup L$ and the two gadget types $B_1$ (blue) and $B_2$ (orange) of the construction, following Donkers and Jansen~\cite{donkersTuringKernelizationDichotomy2021}; the shading matches \cref{fig:necklace_chain}, and leaf blocks are drawn as pentagons pointing toward their free vertex. In~(b) the copy $H_1$ lies on the chaining axis $u$--$s$--$w$ (the order in which $B_1$ is chained later) and $H_2$ branches off at $s = b_2 = c_1$. In~(c) the three leaf blocks lie on the axis $d$--$e$--$f$--$g$ (with $L_3$ the extra copy of $L$, $e = c_3 = c_1$ and $f = b_3 = c_2$) and the subgraphs $R_1$, $R_2$ dangle from $e$, $f$.}
		\label{fig:blocks}
\end{figure}

	Now, we establish important properties about the pathwidth of $L$.
	\begin{claim}
		There is a path decomposition $\LL_{c,b}$ of $L$ that begins with bag $\{c\}$, ends with bag $\{b\}$, and has width at most $\omega$.
	\end{claim}
	\begin{claimproof}
		If $L$ is a tidy leaf block, then there is such a path decomposition of width at most $\pw(H) \leq \omega$ by definition.
		Otherwise, by \cref{thm:pathwidth_starting_ending_decomp} there is such a path decomposition of width at most $\pw(L) + 1 \leq \omega$.
	\end{claimproof}

	Next, we establish a similar claim for $R$.
	\begin{claim}
		There is a path decomposition $\RR_{a,c}$ of $R$ that begins with bag $\{a\}$, ends with bag $\{c\}$, and has width at most $\omega$.
	\end{claim}
	\begin{claimproof}
		By \cref{thm:pathwidth_starting_ending_decomp} such a path decomposition of width at most $\pw(R) + 1$ exists.
		If $b$ is pathwidth critical, then, because $b \notin V(R)$ we have $\pw(R) < \pw(H)$, and therefore $\pw(R) + 1 \leq \pw(H) \leq \omega$.
		Otherwise, if $b$ is not pathwidth critical, then we directly have $\pw(R)+1 \leq \pw(H)+1 = \omega$.
	\end{claimproof}

	Using an analogous proof we also get the following.
	\begin{claim}
		There is a path decomposition $\RR_{c}$ of $R$ of width at most $\omega$, that starts with bag $\{c\}$, and ends with bag $\{c\}$.
	\end{claim}

	We can also obtain a path decomposition of $H$ such that vertex $b$ is in each bag.
	\begin{claim}
		There is a path decomposition $\HH_b$ of $H$ that begins with bag $\{b\}$, ends with bag $\{b\}$, and has width at most $\omega$.
	\end{claim}
	\begin{claimproof}
		We can create such a path decomposition by taking a path decomposition of $H - b$ of width $\pw(H - b)$.
		Then, we add vertex $b$ to all bags, and add a bag $\{b\}$ to the start of the decomposition, and a bag $\{b\}$ to the end of the decomposition.
		Clearly, the resulting decomposition $\HH_b$ has width at most $\pw(H - b) + 1$.
		If $b$ is pathwidth critical for $H$, then $\pw(H - b) + 1 = \pw(H) \leq \omega$.
		Otherwise, $\pw(H - b) + 1 = \pw(H) + 1 = \omega$.
	\end{claimproof}

	\proofsubparagraph{Gadgets of type $\mathbf{1}$.}
	The construction proceeds by defining two types of gadgets.
	The first type of gadget, $B_1$, is created as follows.
	Begin with two disjoint copies of graph $H$; denote them as $H_1$ and $H_2$.
	For $Q \in \{1,2\}$ the copy of $c$ in $H_Q$ is denoted as $c_Q$; the copy of $b$ is denoted as $b_Q$, and the copy of $a$ is denoted as $a_Q$.
	Then, the gadget $B_1$ is created by identifying the vertices $b_2$ and $c_1$; call the resulting vertex $s$.
	Three further vertices of this gadget play a special role later, and are given further names.
	The vertex $a_1$ is called $u$, vertex $b_1$ is called $w$, and vertex $c_2$ is called $t$.
	See \cref{fig:block_B1} for an illustration.

	The construction of gadgets of type $1$ is the only place where our construction differs slightly from the one of Donkers and Jansen: they have an additional copy of the leaf block $L$ attached to vertex $c_2$.
	However, as we show, this additional leaf block is not actually necessary for our result.

	Now, we analyze the pathwidth of gadgets of type $1$.
	\begin{claim}
		Graph $B_1$ has a path decomposition $\BB_1$ that starts with bag $\{u\}$, ends with bag $\{w\}$, and has width at most $\omega$.
	\end{claim}
	\begin{claimproof}
		Recall that $H_1$ contains a leaf block isomorphic to $L$ that contains $w$.
		Call this leaf block $L_1$, and call the graph obtained from $H_1$ by deleting the vertices $V(L_1)\setminus\{c_1\}$ graph $R_1$.
		Clearly, $L_1$ is the ``natural'' copy of $L$ in $H_1$, and $R_1$ the ``natural'' copy of $R$ in $H_1$.

		Let $\RR_1$ be a path decomposition of $R_1$ that begins with bag $\{u\}$ and ends with bag $\{s\}$ of width at most $\omega$.
		This is just the path decomposition $\RR_{a,c}$ where vertices are renamed.
		Let $\HH_2$ be a path decomposition of $H_2$ of width at most $\omega$ that begins with bag $\{s\}$, and ends with bag $\{s\}$.
		This is just the path decomposition $\HH_b$ where vertices are renamed.
		Finally, let $\LL_1$ be a path decomposition of $L_1$ that begins with bag $\{s\}$ and ends with bag $\{w\}$ of width at most $\omega$.
		This is just the path decomposition $\LL_{c,b}$ where vertices are renamed.

		Then, since $s$ is a cut vertex that disconnects $R_1 - s$, $L_1 - s$, and $H_2 - s$, it is easy to confirm that $\RR_1,\HH_2,\LL_1$ is the desired path decomposition $\BB_1$.
	\end{claimproof}

	\proofsubparagraph*{Gadgets of type $\mathbf{2}$.}
	The second type of gadget, $B_2$, has a slightly more complex construction.
	We again start by taking the disjoint union of two copies of $H$, denoted as $H_1$ and $H_2$ with the same convention for $c_Q,b_Q,a_Q$ (for $Q \in \{1,2\}$) we used for gadgets of the first type.
	Then, we add another copy of the graph $L$ to the construction.
	Let us denote the vertex $c$ of this copy as $c_3$ and vertex $b$ of this copy as $b_3$.
	The gadget $B_2$ is then created by identifying the vertices $c_3$ and $c_1$ into a new vertex called $e$ as well as by identifying the vertices $b_3$ and $c_2$ into a new vertex called $f$.
	We again have some additional special names for some further vertices.
	Call vertex $b_1$ vertex $d$ and vertex $b_2$ vertex $g$.
	See \cref{fig:block_B2} for an illustration.

	We again analyze the pathwidth of such a gadget.

	\begin{claim}
		There is a path decomposition $\mathcal{B}_2$ of $B_2$ that starts with bag $\{d\}$, ends with bag $\{g\}$, and has width at most $\omega$.
	\end{claim}
	\begin{claimproof}
		Let $L_1$ be the copy of $L$ in $H_1$ that contains vertex $d$, and $L_2$ be the copy of $L$ in $H_2$ that contains vertex $g$.
		Moreover, let $L_3$ be the copy of $L$ used in the creation of $B_2$ that contains $e$ and $f$.
		Similarly, let $R_1$ be $H_1 - (V(L_1) \setminus \{c_1\})$, and $R_2$ be $H_2 - (V(L_2) \setminus \{c_2\})$.

		Let $\LL_1$ be a path decomposition of $L_1$ starting with bag $\{d\}$, and ending with bag $\{e\}$; let $\LL_2$ be a path decomposition of $L_2$ starting with bag $\{f\}$, and ending with bag $\{g\}$; let $\LL_3$ be a path decomposition of $L_3$ starting with bag $\{e\}$, and ending with bag $\{f\}$, such that each of these three path decompositions has width at most $\omega$.
		These path decompositions exist since each of them is just the decomposition $\LL_{c,b}$ with renamed vertices and a possibly reversed order.

		Similarly, let $\RR_1$ be a path decomposition of $R_1$ starting and ending with bag $\{e\}$, and $\RR_2$ be a path decomposition of $R_2$ starting and ending with bag $\{f\}$, such that both decompositions have width at most $\omega$.
		These decompositions exist since they are just copies of the decomposition $\RR_{c}$ with renamed vertices.

		Then, it is easy to confirm that $\LL_1,\RR_1,\LL_3,\RR_2,\LL_2$ is a path decomposition of $B_2$ of width at most $\omega$.
	\end{claimproof}

	\proofsubparagraph*{The complete construction.}
	
	For each $\ell$ the construction works as follows.
	We create $\ell$ gadgets of type $1$, denoted $B_1^1,\dots,B_1^\ell$, and $\ell - 1$ gadgets of type $2$, denoted $B_2^1,\dots,B_2^{\ell - 1}$.
	For gadget $B_1^i$, let us denote vertex $u$ of the gadget as $u^i$, vertex $s$ as $s^i$, vertex $w$ as $w^i$, and vertex $t$ as $t^i$.
	Similarly, for gadget $B_2^i$ let us denote vertex $d$ as $d^i$, vertex $e$ as $e^i$, vertex $f$ as $f^i$, and vertex $g$ as $g^i$.
	For each $i \in [\ell -1]$, we identify vertex $w^i$ of gadget $B_1^i$ with vertex $d^{i}$ of gadget $B_2^i$; and we identify vertex $u^{i+1}$ of gadget $B_1^{i+1}$ with vertex $g^i$ of gadget $B_2^i$.
	Intuitively, this creates a chain of gadgets $B_1^1,B_2^1,B_1^2,\dots,B_2^{\ell-1},B_1^\ell$, where each two subsequent gadgets share exactly one vertex.
	See \cref{fig:necklace_chain} for an illustration.
	Call the constructed graph $G^\star$.

	\begin{figure}[t]
		\centering
		\includegraphics{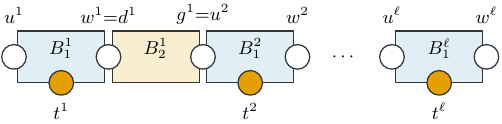}
		\caption{The chain of gadgets forming $G^\star$: gadgets of type $B_1$ are shaded blue and of type $B_2$ orange, and consecutive gadgets share exactly one vertex ($w^i = d^i$, or $g^i = u^{i+1}$), drawn as a circle glued onto both gadgets. The orange vertices form the set $X = \{t^i \mid i \in [\ell]\}$ considered in the blocking-set argument.}
		\label{fig:necklace_chain}
\end{figure}

	\proofsubparagraph*{The pathwidth and treewidth.}
	Now, we bound the pathwidth of $G^\star$.
	\begin{claim}
		We have $\pw(G^\star) \leq \omega$.
	\end{claim}
	\begin{claimproof}
		For each $i \in [\ell]$ and gadget $B_1^i$, let $\BB_1^i$ be a path decomposition of $B_1^i$ of width at most $\omega$ that begins with bag $\{u^i\}$, and ends with bag $\{w^i\}$; this is just the path decomposition $\BB_1$ with renamed vertices.
		Similarly, for each $i \in [\ell-1]$ and gadget $B_2^i$, let $\BB_2^i$ be a path decomposition of $B_2^i$ that begins with bag $\{d^i\}$ and ends with bag $\{g^i\}$ and has width at most $\omega$.
		This is just the path decomposition $\BB_2$ with renamed vertices.

		Then, we have that $\BB_1^1,\BB_2^1,\BB_1^2,\dots,\BB_2^{\ell-1},\BB_1^\ell$ is a path decomposition of $G^\star$ of width at most $\omega$.
		Indeed, for each $i \in [\ell-1]$ the last bag of $\BB_1^i$ is bag $\{w^i\}$, which is identical to the first bag $\{d^i\}$ of $\BB_2^i$ since vertices $w^i$ and $d^i$ were identified.
		Similarly, for each $i \in [\ell-1]$ the last bag of $\BB_2^i$ is bag $\{g^i\}$, which is identical to the first bag $\{u^{i+1}\}$ of $\BB_1^{i+1}$ since $g^i$ and $u^{i+1}$ were identified by the construction.
	\end{claimproof}

	In fact, it is not difficult to show that $\pw(G^\star) = \omega$ (for large enough $\ell$).
	Although this is not necessary to obtain the result, we provide a proof of this fact here since it shows that our analysis of the pathwidth of the construction is tight.

	\begin{claim}
		If $\ell \geq 3$ we have $\pw(G^\star) \geq \omega$.
	\end{claim}
	\begin{claimproof}
		We have $\pw(G^\star) \geq \pw(H)$ since $H$ is a subgraph of $G^\star$.
		Hence, the claim is already proven if $H$ has a tidy leaf block.
		So, assume that $H$ does not have a tidy leaf block.

		Now, consider that vertex $b$ is not pathwidth critical.
		Then, set $H' = H - \{b\}$.
		Observe that $G^\star$ contains three vertex-disjoint copies of $H'$, one in each of three gadgets of type $1$.
		Moreover, by contracting the remaining vertices of the graph, one obtains the following graph as a minor: three copies of $H'$, and an additional vertex that has an edge to each copy.
		Therefore, by basic properties of pathwidth (see, for example,~\cite[Lemma 12.12]{mohring1995Chapter8VLSI}), we have that $\pw(G^\star) \geq \pw(H') + 1 = \pw(H) + 1 = \omega$.
		
		Finally, consider that vertex $b$ is pathwidth critical.
		Then, it must be the case that there is no path decomposition for $L$ of width at most $\pw(H)$ such that the first bag of the decomposition is $\{c\}$ and the last bag is $\{b\}$; indeed otherwise $H$ would have a tidy leaf block.
		By \cref{thm:pathwidth_starting_ending_decomp}, this implies $\pw(L) = \pw(H)$.
		Now, consider the following graph $L^3$: take $3$ vertex-disjoint copies of $L$. 
		Denote the copies of $L$ as $L_1,L_2,L_3$ and denote the respective vertices $c,b$ as $c^1,c^2,c^3$ and $b^1,b^2,b^3$.
		Identify vertices $b^1$ and $c^2$ into vertex $w_\mathrm{L}$, and identify vertices $b^2$ and $c^3$ into vertex $w_\mathrm{R}$.
		Observe that $G^\star$ contains $L^3$ as a minor.
		Indeed, the respective copies of $L$ can be found in the gadgets of type $1$.

		Now, we argue that $\pw(L^3) \geq \pw(L) + 1$, which concludes the proof.
		Consider a path decomposition $X_1,\dots,X_t$ of $L^3$.
		Towards a contradiction, assume it has width $\pw(L)$.
		Since $L^3$ contains $L_1,L_2$, and $L_3$, there is a bag $X_p$ containing $\pw(L) + 1$ vertices of $L_1$, and a bag $X_q$ containing $\pw(L) + 1$ vertices of $L_3$.
		Clearly, $X_p \subseteq V(L_1)$ and $X_q \subseteq V(L_3)$.
		By reversing the path decomposition if necessary we can assume $p < q$.

		The graph $L_\mathrm{L} = L^3[V(L_1) \cup V(L_2)] - w_\mathrm{R}$ is connected since $L$ is biconnected.
		Therefore, all vertices of $L_\mathrm{L}$ are introduced and forgotten before bag $X_q$.
		An analogous argument shows that all vertices of $L_\mathrm{R} =  L^3[V(L_2) \cup V(L_3)] - w_\mathrm{L} $ are introduced and forgotten after bag $X_p$.
		In particular, all vertices of $L_2 \setminus \{w_\mathrm{L},w_\mathrm{R}\}$ are introduced and forgotten after bag $X_p$ but before bag $X_q$.
		
		There is a path from $w_\mathrm{L}$ to a vertex of $X_p$ in $L^3$, and $w_\mathrm{L}$ is the only vertex of $L_2$ of this path.
		Hence, each bag from $X_p$ to the bag introducing $w_\mathrm{L}$ contains a vertex of $L_1$.
		Similarly, each bag from $X_q$ to the bag forgetting $w_\mathrm{R}$ contains a vertex of $L_3$.
		Restrict the path decomposition to the interval $X_p,\dots,X_q$, and replace each vertex of $L_1$ with $w_\mathrm{L}$, and each vertex of $L_3$ with $w_\mathrm{R}$.
		Then, this results in a path decomposition of $L_2$ of width $\pw(L)$, whose first bag contains $w_\mathrm{L}$, and whose last bag contains $w_\mathrm{R}$.
		Thus, we have reached a contradiction.
	\end{claimproof}

	Of course, we have $\tw(G^\star) = \tw(H)$ by the analysis of Donkers and Jansen~\cite{donkersTuringKernelizationDichotomy2021}.
	This holds because each biconnected component of $G^\star$ is a biconnected component of $H$.

	\proofsubparagraph*{The minimal blocking set.}
	Let us now show that this graph contains a minimal $\{H\}$-blocking set of size $\ell$.
	Note that this is not explicitly shown in the paper of Donkers and Jansen~\cite{donkersTuringKernelizationDichotomy2021}.
	Indeed, their paper does not use the language of minimal blocking sets, but it follows easily from their proof.
	First, we need to give a lower bound on the size of optimal solutions.

	The two packings of vertex-disjoint copies of $H$ underlying the argument are illustrated in \cref{fig:packings}, on the developed structure of a representative segment $B_1^i$--$B_2^i$ of the chain.

	\begin{figure}[t]
		\centering
		\begin{subfigure}{\linewidth}
			\centering
			\includegraphics[page=1]{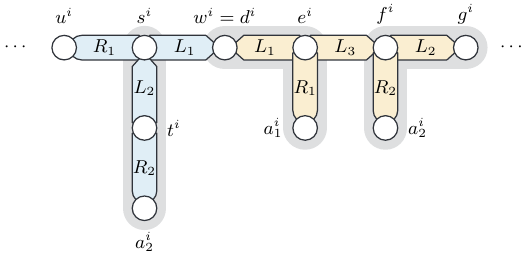}
			\subcaption{A packing of vertex-disjoint copies of $H$ in $G^\star$: each $B_1^i$ contributes one copy (its branch at $s^i$) and each $B_2^i$ two copies, giving $3\ell-2$ in total.}
			\label{fig:packing_opt}
		\end{subfigure}

		\vspace{1em}
		\begin{subfigure}{\linewidth}
			\centering
			\includegraphics[page=2]{figures/packings.pdf}
			\subcaption{A packing in $G^\star - X$ for $X = \{t^i \mid i \in [\ell]\}$: each gadget contributes one copy, giving $2\ell-1$ in total. Red crossed circles mark the former positions of the deleted vertices $t^i$. As $|X| = \ell$, any solution containing $X$ has size at least $3\ell-1 > \opt_\FF(G^\star)$.}
			\label{fig:packing_blocking}
		\end{subfigure}
		\caption{The two packings of vertex-disjoint copies of $H$ underlying the blocking-set argument, each highlighted in gray, on a representative segment $B_1^i$ (blue)--$B_2^i$ (orange) of the chain (continuing on both sides), with the gadget shading as in \cref{fig:necklace_chain}. The construction follows Donkers and Jansen~\cite{donkersTuringKernelizationDichotomy2021}.}
		\label{fig:packings}
\end{figure}

	\begin{claim}
		We have $\opt_{\FF}(G^\star) \geq 3\ell - 2$.
	\end{claim}
	\begin{claimproof}
		It is not difficult to confirm that $G^\star$ contains $3\ell - 2$ vertex-disjoint copies of $H$.
		In each gadget $B_1^i$ consider the copy of $H$ attached at vertex $s^i$, and observe that this copy does not contain any vertex of a different gadget.
		In each gadget $B_2^i$ consider the copy of $H$ containing vertex $d^i$, and the copy of $H$ containing vertex $g^i$.
		This results in a packing of exactly $\ell + 2(\ell-1)$ vertex-disjoint copies of $H$.
		In fact, the packing we describe here is exactly the packing Donkers and Jansen~\cite{donkersTuringKernelizationDichotomy2021} illustrate to prove their property 3 of their Lemma 11.
	\end{claimproof}

	Our eventual goal is arguing that the set $X = \{t^i \mid i \in [\ell]\}$ is a minimal $\{H\}$-blocking set.
	We begin by showing that any proper subset of $X$ is not a $\{H\}$-blocking set.

	\begin{claim}
		For any $v \in X$ there is a set $S_v$ of size $3\ell - 2$ containing $X \setminus \{v\}$ such that $G^\star - S_v$ is $H$-minor free, and such that each connected component of $\prune{|V(L)|}{G^\star - S_v}$ is a minor of $H$.
	\end{claim}
	\begin{claimproof}
		For this claim, let us refer to the proof of Donkers and Jansen~\cite[Lemma 11]{donkersTuringKernelizationDichotomy2021}, and denote the graph they build as $G^{\star\prime}$.
		As elaborated earlier, we actually have that $G^\star$ is a subgraph of $G^{\star\prime}$.

		Recall that $X = \{t^j \mid j \in [\ell]\}$, and fix an arbitrary $v \in X$.
		Donkers and Jansen define a set $S_v'$ of size $3\ell - 1$, which can easily be verified to contain all vertices of $X \setminus \{v\}$.
		This set is such that $G^{\star\prime} - S_v'$ is $H$-minor free, and such that each connected component of $\prune{|V(L)|}{G^{\star\prime}- S_v'}$ is a minor of $H$.
		Slightly closer inspection of their proof reveals that the set they consider contains exactly one vertex which is \emph{not} part of our construction.
		Hence, the set $S_v = S_v' \cap V(G^\star)$ has size only $3\ell - 2$.

		Finally, we establish that each connected component of $\prune{|V(L)|}{G^\star - S_v}$ is a minor of $H$.
		We have that $G^\star - S_v$ is a minor of $G^{\star\prime} - S_v'$.
		By~\cite[Proposition 9]{donkersTuringKernelizationDichotomy2021} we then have that $\prune{|V(L)|}{G^\star - S_v}$ is a minor of $\prune{|V(L)|}{G^{\star\prime} - S_v'}$.
		Since each connected component of $\prune{|V(L)|}{G^{\star\prime} - S_v'}$ is a minor of $H$, and each connected component of $\prune{|V(L)|}{G^\star - S_v}$ is a minor of some connected component of $\prune{|V(L)|}{G^{\star\prime} - S_v'}$, we obtain that each connected component of $\prune{|V(L)|}{G^{\star} - S_v}$ is a minor of $H$.
	\end{claimproof}

	Now, we can show that $X$ is a $\{H\}$-blocking set.
	\begin{claim}
		The set $X$ is a $\{H\}$-blocking set.
	\end{claim}
	\begin{claimproof}
		Set $Z = G^\star - X$.
		We will argue that $Z$ contains $2\ell - 1$ vertex-disjoint copies of $H$.
		Then, it is clear that any solution containing $X$ must have size at least $2\ell - 1 + |X| = 3\ell-1$.

		We describe the packing next.
		In each gadget $B_1^i$ we pack the copy of $H$ that contains vertex $u^i$.
		In each gadget $B_2^i$ we pack the copy of $H$ that contains the vertices $e^i$ and $f^i$.
		This packed graph contains no vertices of a different gadget.
		Clearly, we pack exactly one copy of $H$ per gadget, and hence the packing has size $2\ell-1$.

		This packing is again similar to a packing described by Donkers and Jansen in their proof of Lemma 11.
		Concretely, in the proof of property 2 of their lemma, they provide a packing of $3\ell-1$ vertex-disjoint subgraphs of their construction, where $\ell$ graphs contain a vertex of the leaf blocks we do not use in our construction, and the same $\ell$ graphs contain one vertex of $X$ each.
		Hence, when omitting these graphs from the packing, one obtains the desired packing we sketched above.
	\end{claimproof}
	Combining the claims above concludes the proof.
\end{proof}

\subsection{Families of Biconnected Graphs}
\label{sec:biconnected_families}

Next, we consider the case that all graphs of $\FF$ are biconnected on at least three vertices.
In this case, we can show that there is no polynomial kernel when parameterizing by the distance to graphs with pathwidth $\min_{H \in \FF} \pw(H)$.
On a technical level, this result is obtained by combining the result that necklaces have unbounded minimal extended $\FF$-blocking sets (when minor-minimal $\FF$ contains only biconnected graphs on at least three vertices) with the observation that there are arbitrarily long $\{H\}$-necklaces with the same pathwidth as $H$ for all $H \in \FF$.

\thmBiconnectedPW*{}
\begin{proof}
	Let $F \in \FF$ be a graph of $\FF$, and $X_1,\dots,X_t$ be a path decomposition of $F$ of width $\pw(F)$.
	Without loss of generality, $X_1 = \{v\}$ and $X_t = \{w\}$ for some $v,w \in V(F)$ (we may have $v = w$).

	Fix an arbitrary integer $\ell \geq 1$.
	We create a graph $G_\ell^F$ as follows.
	Take $\ell$ disjoint copies of $F$, and denote them as $F_1,\dots,F_\ell$.
	Let us denote vertex $v$ of copy $F_i$ as $v_i$, and similarly let us denote vertex $w$ of copy $F_i$ as $w_i$.
	For each $i \in [\ell -1]$ add an edge from $w_i$ to $v_{i+1}$.
	See \cref{fig:necklace_GlF} for an illustration of $G_\ell^F$.

	\begin{figure}[t]
		\centering
		\includegraphics{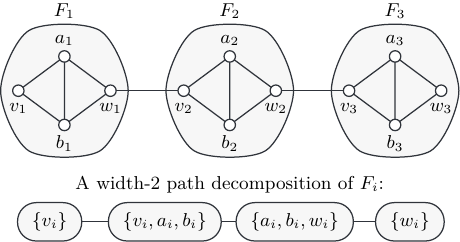}
		\caption{The $\{F\}$-necklace $G_3^F$ for the diamond $F$, obtained from $K_4$ by deleting one edge.
		In each copy $F_i$, vertices $v_i$ and $w_i$ are the two nonadjacent vertices; the edges $w_1v_2$ and $w_2v_3$ join consecutive copies.
		The displayed width-$2$ path decomposition of a copy $F_i$ starts with bag $\{v_i\}$ and ends with bag $\{w_i\}$.
		Concatenating these decompositions and adding $w_i$ to the first bag $\{v_{i+1}\}$ of the next copy, for $i \in \{1,2\}$, gives a path decomposition of $G_3^F$ of width $2$.}
		\label{fig:necklace_GlF}
\end{figure}

	It is easy to see that $G_\ell^F$ is an $\FF$-necklace of length $\ell$, and therefore by \cref{thm:biconnected_graphs_large_extended_blocking_sets}, it has a good minimal extended $\FF$-blocking set $(X,Y)$ with $|X| = \ell$ and $|Y| = 2$.

	The property about sets $S$ making $G_\ell^F$ free of an $F$-minor is also easy to see, because each such set $S$ must contain a vertex from each of the $\ell$ copies of $F$, and this immediately yields the stated property about biconnected components of $G_\ell^F - S$.

	To conclude the proof, it remains to argue that the pathwidth of $G_\ell^F$ is at most $\pw(F)$.
	Let $F_1,\dots,F_\ell$ be the $\ell$ disjoint copies of $F$ within $G_\ell^F$.
	For each $i \in [\ell]$ and each $j \in [t]$ let $X_j^i$ be the bag $X_j$, where the vertices of $X_j$ are replaced with their respective copy within graph $F_i$.
	Then, 
	\[\XX_0 = X_1^1,\dots,X_t^1,X_1^2,\dots,X_t^2,\dots,X_1^\ell,\dots,X_t^\ell\] is clearly a path decomposition of the graph obtained by taking the disjoint union of $F_1,\dots,F_\ell$ of width $\pw(F)$; indeed this sequence of bags is obtained from appending path decompositions of the graphs $F_i$ of width $\pw(F)$.
	We modify this sequence $\XX_0$ by adding the vertex $w_{i}$ to the bag $X_1^{i+1}$ for each $i \in [\ell-1]$.
	Then, this results in a path decomposition $\XX$ of $G_\ell^F$ since for each $i \in [\ell-1]$ the endpoints of the edge from $F_i$ to $F_{i+1}$ are now contained in bag $X_1^{i+1}$.
	Moreover, since $|X_1| = 1$, each bag $X_1^i$ contains at most $2$ vertices.
	As $F$ is connected with at least two vertices, $\pw(F) \geq 1$, and therefore the width of $\XX_0$ is the same as the width of $\XX$.
	Hence, $\pw(G_\ell^F) = \pw(F)$, concluding the proof.
\end{proof}

\subsection{Families Containing Multiple Graphs}
\label{sec:pw_lb_families}

After having \cref{thm:biconnectedPW,thm:nonBiconnectedPW} at hand, we want to extend these results to families $\FF$ containing multiple graphs that are not necessarily biconnected.
For the result of Donkers and Jansen~\cite{donkersTuringKernelizationDichotomy2021}, they show that one can choose a graph $H \in \FF$ with $\tw(H) = \min_{F \in \FF} \tw(F)$ such that the construction for the single graph $H$ essentially works for the full family.
However, because the construction of \cref{thm:nonBiconnectedPW} for a graph $H$ can have pathwidth $\pw(H)+1$, we cannot get the same exact result here with treewidth replaced with pathwidth.
Nevertheless, we obtain the following \namecref{thm:pwModulatorBlockingSet}.

\begin{restatable}{lemma}{thmPathwidthBlockingSet}
	\label{thm:pwModulatorBlockingSet}
	Let $\FF$ be a minor-minimal family of connected graphs on at least three vertices each, and $H$ be a graph of $\FF$ that has a smallest leaf block among all graphs of $\FF$.
	Moreover, let $\ell \geq 1$ be an arbitrary integer.

	If $H$ is biconnected, then there is a graph $G_\ell$ with a good minimal extended $\FF$-blocking set $(X,Y)$ with $|X| = \ell$ and $|Y| = 2$, and $\pw(G_\ell) = \pw(H)$.

	If $H$ is not biconnected, then there is a graph $G_\ell$ with a minimal $\FF$-blocking set of size $\ell$.
	Moreover, the pathwidth of $G_\ell$ is at most $\pw(H) + 1$,
	and even at most $\pw(H)$ if one of the smallest leaf blocks of $H$ is tidy.
\end{restatable}
\begin{proof}
	We use a different construction depending on whether $H$ is biconnected.
	\proofsubparagraph*{$H$ is biconnected.}

	If $H$ is biconnected, $G_\ell$ is the graph $G_\ell^H$ we obtain by applying \cref{thm:biconnectedPW} with family $\{H\}$.
	Then, $G_\ell$ has a good minimal extended $\{H\}$-blocking set $(X,Y)$ with $|X| = \ell$ and $|Y| = 2$.
	We need to argue that $(X,Y)$ is also a good minimal extended $\FF$-blocking set.
	
	Consider an arbitrary pair $(X',Y')$ with $X' \subseteq X$ and $Y' \subseteq Y$, such that at least one of the two inclusions is proper.
	Since $(X,Y)$ is a minimal extended $\{H\}$-blocking set of $G_\ell$, there is a set $S$ such that $G_\ell - S$ is $\{H\}$-minor free, $|S| = \opt_{\{H\}}(G_\ell)$, $S$ contains all vertices of $X'$, and no two vertices of $Y'$ are connected in $G_\ell - S$.
	Moreover, by \cref{thm:biconnectedPW}, each biconnected component of $G_\ell - S$ contains fewer than $|V(H)|$ vertices.
	Let $\LL$ be the set of leaf blocks of graphs in $\FF$.
	Then, since $H$ is biconnected and has the smallest leaf block of any graph of $\FF$, we know that $G_\ell - S$ is $\LL$-minor free.
	Clearly, this implies that $G_\ell - S$ is $\FF$-minor free.
	Hence, we have that $\opt_{\{H\}}(G_\ell) = \opt_{\FF}(G_\ell)$.

	Now, observe that since $(X,Y)$ is a minimal extended $\{H\}$-blocking set, any set $S$ such that $G_\ell - S$ is $\FF$-minor free that contains $X$ and disconnects the vertices $Y \setminus S$ has size larger than $\opt_{\{H\}}(G_\ell) = \opt_{\FF}(G_\ell)$.
	Hence, $(X,Y)$ is an extended $\FF$-blocking set.
	By the existence of the solution described in the previous paragraph, it is also a good minimal extended $\FF$-blocking set.
	The pathwidth bound also follows from \cref{thm:biconnectedPW}.

	\proofsubparagraph*{$H$ is not biconnected.}
	Next, consider that $H$ is not biconnected, and let $L$ be a smallest leaf block of $H$, where $L$ is a tidy leaf block if possible.
	In this case, we use the construction of \cref{thm:nonBiconnectedPW}. 
	That is, $G_\ell = G_\ell^\star$ of \cref{thm:nonBiconnectedPW} used with graph $H$.

	The pathwidth bound is explicitly stated in the lemma.
	So, it remains to show that $G_\ell$ has a minimal $\FF$-blocking set of size $\ell$.

	By \cref{thm:nonBiconnectedPW}, $G_\ell$ has a minimal $\{H\}$-blocking set $X$ of size $\ell$.
	Moreover, for each $v \in X$ there is a set $S_v$ of size $\opt_{\{H\}}(G_\ell)$ containing $X \setminus \{v\}$, such that $G_\ell - S_v$ is $\{H\}$-minor free, and each connected component of $\prune{|V(L)|}{G_\ell - S_v}$ is a minor of $H$.
	Consider an arbitrary $H' \in \FF$, and assume that $H'$ is a minor of $G_\ell - S_v$.
	By our choice of $H$, each leaf block of $H'$ has at least $|V(L)|$ vertices.
	Now, by an iterated application of \cref{thm:kill_minors_by_leaf_blocks} we find that $H'$ must be a minor of a connected component of $\prune{|V(L)|}{G_\ell - S_v}$: leaf blocks of $G_\ell - S_v$ that contain fewer than $|V(L)|$ vertices clearly do not have any leaf block of $\FF$ as a minor.
	Since each connected component of $\prune{|V(L)|}{G_\ell - S_v}$ is a minor of $H$, we must have that $H'$ is a minor of $H$.
	If $H' \neq H$, this contradicts that $\FF$ is minor-minimal.
	If $H' = H$, this contradicts that $G_\ell - S_v$ is $\{H\}$-minor free.
	Hence, we have reached a contradiction, and we obtain that $G_\ell - S_v$ is $\FF$-minor free.
	In particular, the argument shows that $\OPT_{\FF}(G_\ell) = \OPT_{\{H\}}(G_\ell)$.

	We also know that $X$ is a $\{H\}$-blocking set of $G_\ell$.
	Hence, there is no set $S \subseteq V(G_\ell)$ of size $\OPT_{\{H\}}(G_\ell)$ such that $G_\ell - S$ is $\FF$-minor free and $X \subseteq S$.
	From the arguments above, we now obtain that $X$ is a minimal $\FF$-blocking set.
\end{proof}

Let us now extract two interesting corollaries from \cref{thm:pwModulatorBlockingSet} and the fact that for $\FF = \{K_2\}$ the problem \fDeletion{} has no polynomial kernel when parameterized by the distance to graphs with pathwidth at most $2$~\cite{cyganHardnessLosingWidth2014}.

\thmPWMaxPlusOneLB*{}
\begin{proof}
	If $K_2 \in \mathcal{F}$, then the minor-minimality of $\FF$ implies $\FF = \{K_2\}$, and the result
	follows from~\cite{cyganHardnessLosingWidth2014}.
	Otherwise, 
	the statement follows directly from \cref{thm:pwModulatorBlockingSet,thm:lb_from_extended_blocking_sets}.
\end{proof}

Now, for a second corollary, and in fact the main motivation for the considered parameterization, we consider the \pwEtaDeletionEtaPlusOne{} problem.

\thmPWEtaEtaPlusOne*{}
\begin{proof}
	Recall that \pwEtaDeletion{} is an \fDeletion{} problem, where $\FF$ contains all minor-minimal graphs with pathwidth $\eta + 1$.
	Moreover, it is well-known 
	that $\FF$ contains a tree (start from a tree with pathwidth $> \eta$ and contract edges until it becomes a minor-minimal graph with pathwidth $\eta+1$).
	Hence, there is a graph $H$ with a smallest leaf block of $\FF$ that is a tree, and this leaf block $L$ is isomorphic to $P_2$.
	Moreover, it is clear that all vertices of $H$ are pathwidth critical, and that $L$ is a tidy leaf block.
	Hence, \cref{thm:pwModulatorBlockingSet} provides, for each integer $\ell \geq 1$, a graph $G_\ell$ with a minimal $\FF$-blocking set of size $\ell$, such that $\pw(G_\ell) \leq \pw(H) = \eta + 1$.
	Then, the result follows from \cref{main_thm:blocking_sets_kernel_lb}.
\end{proof}
 
\section*{AI Disclosure}

We used generative AI tools, including models by Anthropic and OpenAI such as Claude Fable 5, Claude Opus 4.7, GPT-5.6 Sol, and GPT-6 Astra.
These tools helped us proofread the paper, improve the exposition, and search the literature.
Moreover, we used them to review mathematical content and suggest corrections, simplifications, and minor generalizations.
They also assisted in preparing the TikZ code used for the figures of the paper.
We verified all incorporated suggestions and take full responsibility for the content of the article.

\newpage
\bibliography{bib}

\end{document}